\documentclass[twoside,11pt]{article}
\usepackage{jair, rawfonts, }

\usepackage{epsfig}
\usepackage{fancyhdr}
\usepackage{amsmath}

\usepackage{amsfonts}
\usepackage{array}
\usepackage{graphicx}
\usepackage{url}
\usepackage{bm}
\usepackage{breqn}
\usepackage{xcolor}
\usepackage{soul}
\usepackage{amssymb}
\usepackage[breaklinks=true,hidelinks]{hyperref}
\usepackage{breakcites}
\usepackage{algorithm}
\usepackage{algpseudocode}
\usepackage{multirow}
\usepackage[justification=centering]{caption}
\usepackage{float}
\usepackage{caption}
\usepackage{graphics}
\usepackage{mathtools}
\usepackage{amsthm}

\usepackage{subcaption}
\usepackage{hyperref}

\newcolumntype{P}[1]{>{\centering\arraybackslash}p{#1}}
\newcolumntype{M}[1]{>{\centering\arraybackslash}m{#1}}

\newtheorem{theorem}{Theorem}[section]

\jairheading

\begin{document}

\sloppy

\title{TAD-SIE: Sample Size Estimation for Clinical Randomized Controlled Trials using a Trend-Adaptive Design with a Synthetic-Intervention-Based Estimator}

\author{\name Sayeri Lala \email slala@princeton.edu \\
		\addr Department of Electrical and Computer Engineering, Princeton University \\
		Princeton, NJ 08544 USA \\
       \name Niraj K. Jha \email jha@princeton.edu \\
       \addr Department of Electrical and Computer Engineering, Princeton University \\
       Princeton, NJ 08544 USA  
       }
       


\maketitle

\begin{abstract}
Phase-3 clinical trials provide the highest level of evidence on drug safety and effectiveness needed for market
approval by implementing large randomized controlled trials (RCTs). However, 30-40\% of these trials fail mainly because such studies have inadequate sample sizes, stemming from the inability to obtain accurate initial estimates of average treatment effect parameters. To remove this obstacle from the drug development cycle, we present a new algorithm called Trend-Adaptive Design with a Synthetic-Intervention-Based Estimator (TAD-SIE) that appropriately powers a parallel-group trial, a standard RCT design, by leveraging a state-of-the-art hypothesis testing strategy and a novel trend-adaptive design (TAD). Specifically, TAD-SIE uses SECRETS (Subject-Efficient Clinical Randomized Controlled Trials using Synthetic Intervention) for hypothesis testing, which simulates a cross-over trial in order to boost power; doing so, makes it easier for a trial to reach target power within trial constraints (e.g., sample size limits). To estimate sample sizes, TAD-SIE implements a new TAD tailored to SECRETS given that SECRETS violates assumptions under standard TADs. In addition, our TAD overcomes the ineffectiveness of standard TADs by allowing sample sizes to be increased across iterations without any condition while controlling significance level with futility stopping. Our TAD also introduces a hyperparameter that 
enables trial designers to trade off the final sample size with the number of iterations taken by the algorithm in order 
to implement a sample-efficient or time-efficient solution under given resource constraints. On a real-world 
Phase-3 clinical RCT (i.e., a two-arm parallel-group superiority trial with an equal number of subjects per arm), 
TAD-SIE reaches typical target operating points of 80\% or 90\% power and 5\% significance level in contrast to 
baseline algorithms that only get at best 59\% power and 4\% significance level. For a target power of 80\% and 
target significance level of 5\%, the sample-efficient mode of TAD-SIE requires a median of 300-400 subjects and a 
median of 4-5 iterations while the time-efficient mode of TAD-SIE requires a median of 500-700 subjects and a median 
of 2 iterations, demonstrating the advantages of each mode (similar advantages were observed for a target power of 
90\%). While TAD-SIE can only be used for trials with rapidly measurable primary outcomes (i.e., the outcome defining 
the average treatment effect) due to its sequential nature, we plan to address this limitation in future work by 
developing a scheme that fits an estimator that predicts the primary outcome from different but correlated and rapidly 
measurable outcomes in order to avoid directly measuring the primary outcome. 
\end{abstract}

\providecommand{\keywords}[1]
{
  \small	
  \textbf{\textit{Keywords---}} #1
}


\section{Introduction}

Clinical trials are performed in phases in order to characterize safety and effectiveness of new drugs prior to market 
approval. In particular, while Phase-1 and Phase-2 trials gather initial evidence on drug effects using small sample 
sizes, typically less than 100 subjects per arm, in order to identify tolerable and potentially effective dosages, 
Phase-3 trials establish safety and effectiveness over candidate dosages by using Randomized Controlled Trials
(RCTs) with large sample sizes, typically hundreds to thousands of subjects 
\cite{friedman2015fundamentals,stanley2007design}. Despite Phase-3 trials accounting for 60\% of R\&D investment 
for clinical trials (approximately $\$500$ million USD per drug in Year 2019), 30-40\% of these trials fail to 
proceed to market approval \cite{wong2019estimation,harrer2019artificial}, primarily because they fail to establish 
effectiveness \cite{hwang2016failure}; in other words, these studies lack appropriate sample sizes to reach target 
power and target significance level \cite{friedman2015fundamentals}. 

To ensure that any Phase-3 trial succeeds and thereby make the drug development process more efficient, we present a 
new solution called Trend-Adaptive Design with a Synthetic-Intervention-Based Estimator (TAD-SIE) that implements 
a parallel-group RCT, a standard design \cite{friedman2015fundamentals}, of adequate sample size in order to reach 
target power and target significance level. To do this, TAD-SIE adopts a state-of-the-art hypothesis testing scheme 
called SECRETS (Subject-Efficient Clinical Randomized Controlled Trials using Synthetic Intervention) 
\cite{lala2023secrets}, which simulates a cross-over trial \cite{friedman2015fundamentals} in order to increase 
power over standard hypothesis testing; using SECRETS increases the chance of reaching target power within a 
trial's sample size constraints, especially under noisy estimates of the average treatment effect (ATE) 
\cite{friedman2015fundamentals,yao2021survey} and variance of the primary outcome, i.e., the outcome defining the 
ATE. Then, to estimate the sample size required for SECRETS, TAD-SIE implements a novel trend-adaptive design (TAD) 
that gradually increases the sample size while refining estimates of the ATE and outcome variance under SECRETS 
based on accumulated RCT data in order to converge to a sample size that yields the target power; meanwhile, the 
algorithm controls significance level by incorporating futility stopping. Given its iterative nature, TAD-SIE can 
reach the target operating point by trading off final sample size with number of iterations; therefore, we present 
sample-efficient and time-efficient modes for TAD-SIE that can be chosen based on trial constraints.    
 
Specifically, our contributions are as follows:
\begin{enumerate}
	\item We present TAD-SIE, a novel framework that solves the problem of sample size estimation for clinical 
parallel-group RCTs by using a new trend-adaptive algorithm tailored to a powerful hypothesis testing scheme, SECRETS.   
	\item We validate TAD-SIE on a real-world Phase-3 clinical RCT under a standard setup (i.e., a two-arm 
superiority trial with an equal number of subjects per arm \cite{friedman2015fundamentals}), showing that TAD-SIE 
reaches typical target operating points (80\% or 90\% power and 5\% significance level), unlike standard baseline 
algorithms. In addition, we show that the sample-efficient and time-efficient modes of TAD-SIE can effectively 
trade off final sample size with number of iterations, thereby making TAD-SIE conducive to resource constraints. 
	\item We validate the premises underlying TAD-SIE with ablation studies. 
\end{enumerate}
  
The rest of the article is organized as follows. We provide background on prior approaches and topics relevant to our 
framework in Section \ref{sec:background} and then present the framework in Section \ref{sec:method}. We explain how 
we evaluate performance in Section \ref{sec:eval}. We present our results in Section \ref{sec:results}, discuss their 
implications in Section \ref{sec:discuss}, and draw conclusions in Section \ref{sec:con}.          
  
\section{Background} \label{sec:background}

In this section, we provide background on concepts relevant to understanding TAD-SIE, which include methods for 
initial sample size estimation, adaptive designs that can adjust sample size estimates, and a powerful hypothesis 
testing algorithm, called SECRETS.

\subsection{Initial Sample Size Estimation}

Trial designers calculate the sample size required for an RCT by using the hypothesis testing framework 
\cite{rosner2015fundamentals}. For example, for a two-arm, parallel group, superiority trial with an equal number 
of participants per arm, a common type of RCT design \cite{friedman2015fundamentals}, the sample size formula is 
determined by a two-sample, two-sided hypothesis test over the value of the ATE $\delta_{true}$, where the null 
hypothesis ($H_{0}$) posits that $\delta_{true}=0$  and the alternative hypothesis ($H_{1}$) posits 
$\delta_{true}=\delta$, where $\delta>0$. Under this setup, the sample size required to obtain a power of 
$1-\beta_{target}$ (\%) and significance level of $\alpha_{target}$ (\%) is given by 
Eq.~(\ref{eq:sample_size_two_sample}), where $n_{a}$ is the arm size, $\sigma_{control}^2$ and $\sigma_{treat}^2$ are the variances of the primary outcome under the control and treatment groups, respectively, and $z_{1-\alpha_{target}/2}$ and $z_{1-\beta_{target}}$ correspond to the $(1-\alpha_{target}/2)$th and $(1-\beta_{target})$th percentiles of a standard normal distribution.

\begin{equation}
    \label{eq:sample_size_two_sample}
    n_{a}=\frac{(\sigma_{control}^{2}+\sigma_{treat}^{2})(z_{1-\alpha/2}+z_{1-\beta})^{2}}{\delta^{2}}
\end{equation}

Trial designers set the target power and target significance level, with power typically above 80\% and
significance level below 5\%, and can also set $\delta$ or the standardized treatment effect (i.e., the ratio of
ATE to standard deviation) to the smallest clinically meaningful value that is of interest to detect in the
absence of any estimates of these parameters \cite{hulley2013designing}. Otherwise, $\delta$,
$\sigma_{control}^2$, and $\sigma_{treat}^2$ can be estimated from prior studies (e.g., RCTs or external pilot studies) if available or internal pilot studies conducted by trial designers \cite{friedman2015fundamentals,hulley2013designing}; however, calculations based on these data sources, especially pilot studies, tend to overestimate treatment effects and underestimate variances since the data fail to capture the distribution of the ATE over the target population, thereby resulting in underpowered studies \cite{friedman2015fundamentals,teare2014sample}. 

\subsection{Adaptive Designs}

Adaptive trial designs have been proposed as a way to calculate appropriate sample sizes, given noisy initial sample 
size estimates. Specifically, adaptive designs adjust initial estimates based on interim analyses conducted over the course of the trial (up till reaching the initial planned sample size), where at each interim analysis, the outcome data from participants entered thus far into the trial are reviewed. Standard approaches for doing so include Group Sequential Designs (GSD), Stochastic Curtailment, and TADs.  

A GSD can decrease the initial planned sample size through repeated testing, i.e., testing at each interim analysis. 
Specifically, a GSD sequentially acquires RCT data from groups of subjects at each interim analysis and performs 
statistical testing at each analysis based on all the RCT data acquired thus far. In order to control 
the significance level, the test boundary used at each analysis, especially any preceding the final one, needs
to be increased relative to the test boundary used by a nonadaptive trial; for example, for a two-sided test
with 5\% significance level, a GSD with two planned analyses would use a critical value of 2.78 at the first
analysis and 1.97 at the second analysis under an O'Brien-Fleming model \cite{gsd_boundary}, compared to a
conventional value of 1.96 used by a nonadaptive design \cite{friedman2015fundamentals}. Consequently, a GSD
yields marginal reductions in sample sizes when the standardized treatment effect is larger than that used for
planning \cite{mehta2011adaptive}. For GSDs to be useful, a trial designer still needs to have a good prior over
the range of the standardized treatment effect, otherwise risks underpowering or overpowering the study \cite{mehta2011adaptive}, where the latter consequence is not ideal since it consumes substantially more resources than necessary.    

Stochastic Curtailment is another approach that can decrease the initial planned sample size by terminating trials 
that appear futile, that is, likely to fail in reaching statistical significance 
\cite{friedman2015fundamentals,snapinn2006assessment}. It determines futility based on Conditional Power (CP),
which extrapolates power at the final sample size (i.e., the initial estimate of the sample size) conditioned on
the value of the interim test statistic. For example, for a one-sided hypothesis test with significance level
$\alpha$, CP is determined by Eq.~(\ref{eq:cp_formula}) \cite{lan1988b}, where $t$ is the information fraction,
defined as the ratio of the sample size at interim analysis to the planned sample size, $z$ is the test
statistic calculated from interim data, $Z$ is the random variable corresponding to the test statistic at
the final sample size, $z_\alpha$ is the test boundary given by the $\alpha$th percentile under the standard
normal distribution, $\delta$ is the ATE to test against under $H_1$, and $\Phi$ is the cumulative distribution
function for the standard normal distribution. If CP at any interim analysis lies below some pre-specified
futility threshold, the trial is terminated. While Stochastic Curtailment can be used to control the
significance level and reduce resources expended by terminating early \cite{snapinn2006assessment}, it reduces
power \cite{snapinn2006assessment} and therefore cannot be used to appropriately power studies.    

\begin{equation} 
\label{eq:cp_formula}
CP(t,z)= \text{Pr} (Z > z_{\alpha}|t,z,\delta)=1-\Phi\biggr[\frac{z_{\alpha} - z\sqrt{t}-\delta(1-t)}{\sqrt{1-t}}\biggr]
\end{equation}

Instead of decreasing the initial planned sample size, TADs can increase it based on trends observed from
interim data. Among trend-adaptive algorithms, those based on CP have been recommended since they can control
significance level without making statistical adjustments to the test statistic and test critical value
\cite{friedman2015fundamentals,mehta2011adaptive,chen2004increasing}. They do this by only permitting sample
size increases when the trend in the data appears ``promising," a condition determined by CP at interim
analysis. In particular, CP is evaluated based on the data at information fraction $t$, i.e.,
$\delta=\hat{\delta}=z/\sqrt{t}$, which simplifies Eq.~(\ref{eq:cp_formula}) to Eq.~(\ref{eq:cp_formula_trend})
\cite{lan1988b}, where the second line follows from the relationship that $1-\Phi(x)=\Phi(-x)$. If CP is
sufficiently high, i.e., 50\% \cite{chen2004increasing}, or lies within a promising range, where the lowerbound
is determined by the sample size at interim analysis, the initial planned sample size, the maximum sample size,
and target power, and the upperbound is given by target power \cite{mehta2011adaptive}, the final sample size
can be increased according to the sample size formula based on revised parameter estimates obtained from interim
data or according to CP by setting CP to the target power and solving it for sample size, given the revised
parameter estimates. In practice, such TADs have marginal impact on increasing power since the probability of satisfying the CP criterion at an interim analysis remains low \cite{mehta2011adaptive}.

\begin{align} \label{eq:cp_formula_trend}
CP(t,z) = \text{Pr} (Z > z_{\alpha}|t,z,\hat{\delta}) &=1-\Phi\biggr[ \frac{z_{\alpha}}{\sqrt{1-t}} - \frac{z}{\sqrt{t(1-t)}} \biggr] \\
&= \Phi\biggr[\frac{z}{\sqrt{t(1-t)}} - \frac{z_{\alpha}}{\sqrt{1-t}} \biggr]  \nonumber
\end{align}

\subsection{SECRETS} 

SECRETS \cite{lala2023secrets}, illustrated in Fig.~\ref{fig:secrets}, is a hypothesis testing algorithm that increases power of an already-conducted parallel-group RCT by simulating the cross-over design \cite{friedman2015fundamentals} with a state-of-the-art counterfactual estimation algorithm called Synthetic Intervention (SI) \cite{agarwal2023synthetic}. We present the framework for a two-arm superiority trial with an equal number of participants per arm since this is a common setup \cite{friedman2015fundamentals}. 

\begin{figure}[hbt!]
    \centering
\makebox[\textwidth][c]{\includegraphics[scale=0.6]{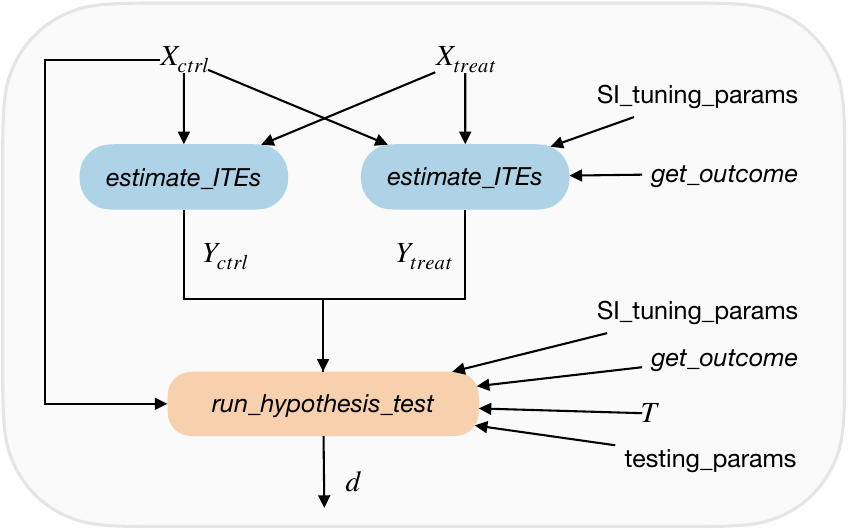}}
    \caption{Flowchart of SECRETS. Note that both calls to \emph{estimate\_ITEs} take in
SI\_tuning\_params and \emph{get\_outcome}, but we have omitted the arguments for brevity.}
    \label{fig:secrets}
\end{figure}

First, it estimates individual treatment effects (ITEs) across all patients in the control and treatment groups,
given by $Y_{ctrl}$ and $Y_{treat}$, respectively, using the observed patient data, $X_{ctrl}$ and $X_{treat}$, and procedure \emph{estimate\_{ITEs}}, which takes in hyperparameters for SI, contained in the SI\_tuning\_params
dictionary, and function \emph{get\_outcome}, which calculates the primary outcome of the trial given a patient's response trajectory under some intervention. Then, it performs hypothesis testing using the merged set of ITEs with \emph{run\_hypothesis\_test}, which implements a data-driven hypothesis test procedure given that the ITEs calculated under \emph{estimate\_{ITEs}} do not satisfy an assumption under standard hypothesis testing. To implement the new test procedure, \emph{run\_hypothesis\_test} uses hyperparameters contained in SI\_tuning\_params and testing\_params, hyperparameter $T$, and function \emph{get\_outcome}. The procedures \emph{estimate\_{ITEs}} and \emph{run\_hypothesis\_test} are described in more detail next.
 
 \begin{figure*}[hbt!]
\makebox[\textwidth][c]{\includegraphics[scale=0.6]{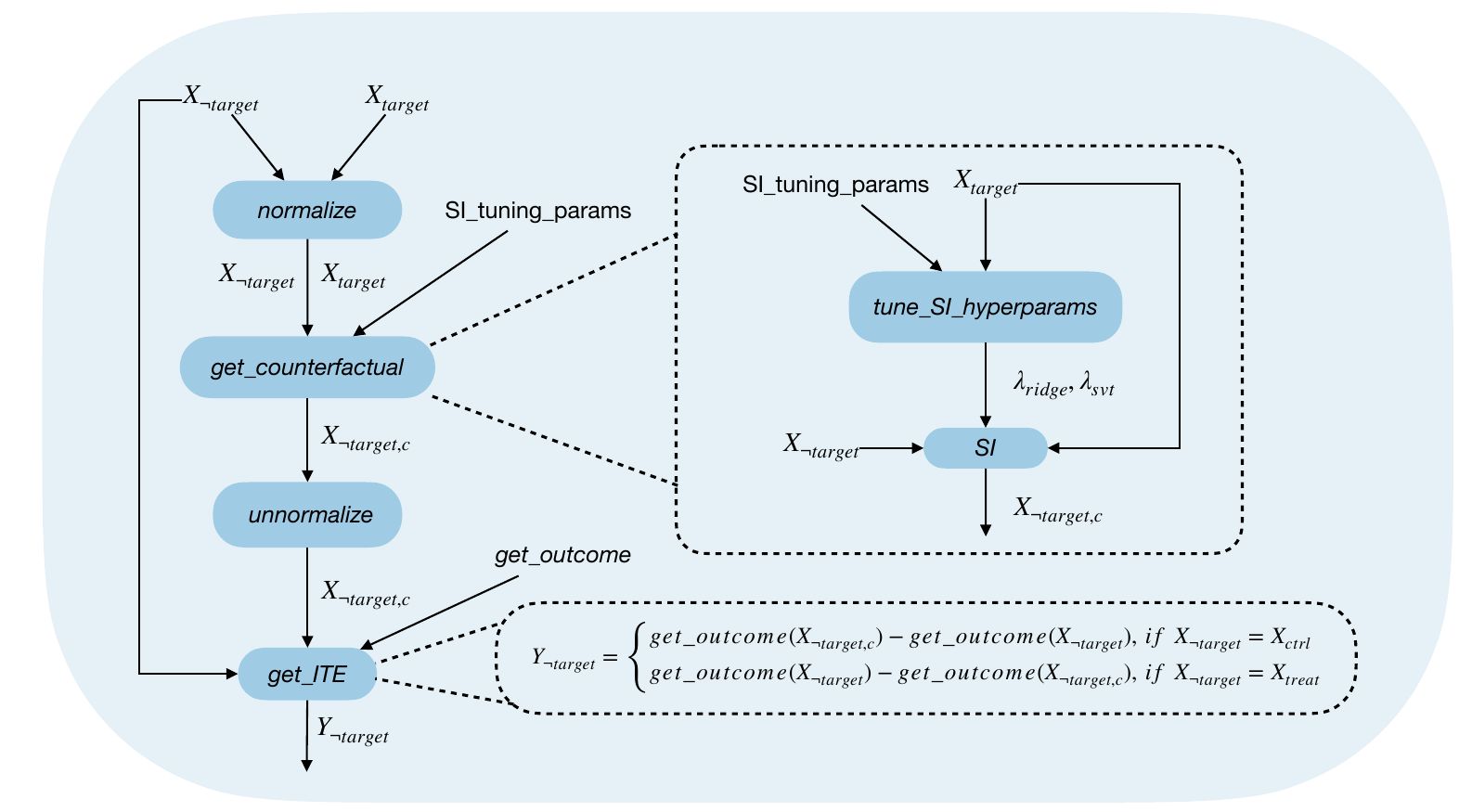}}
    \caption{Flowchart of \emph{estimate\_ITEs}.}
    \label{fig:est_ites}
\end{figure*}

The flowchart of \emph{estimate\_{ITEs}} is shown in Fig.~\ref{fig:est_ites}. To estimate an ITE per
patient, \emph{estimate\_{ITEs}} calculates the counterfactual outcome under the target intervention to
which the patient was not exposed and then takes the difference between the counterfactual outcome under the
target intervention and the observed outcome under the intervention to which the patient was exposed. To do
this, \emph{estimate\_{ITEs}} takes in $X_{\neg target}$, the data of the group unexposed to the target
intervention, and $X_{target}$, the data of the group exposed to the target intervention, normalizes the data,
and calls \emph{get\_counterfactual} to calculate $X_{\neg target,c}$, the counterfactual outcome under the
target intervention per subject from the unexposed group. \emph{get\_counterfactual} uses SI to do this,
which calculates the counterfactual outcome under the target intervention for some target unit using observed
outcomes from a set of donor units exposed to the target intervention. Specifically,
\emph{get\_counterfactual} tunes SI's regularization hyperparameters with \emph{tune\_SI\_hyperparams},
which takes in $X_{target}$ and SI\_tuning\_params (this contains $r_{train,val}$, a hyperparameter for the
ratio of the training to validation split), and runs SI with the tuned hyperparameters and $X_{target}$ as the
donor data and each unit from $X_{\neg target}$ as target unit data. After unnormalizing the observed and
counterfactual outcomes, \emph{estimate\_{ITEs}} calls \emph{get\_ITE} to calculate the ITE using
$X_{\neg target}$, $X_{\neg target,c}$, and \emph{get\_outcome}. \emph{estimate\_{ITEs}} is used to
calculate $Y_{ctrl}$ by setting $X_{\neg target}$ to $X_{ctrl}$ and $X_{target}$ to $X_{treat}$ and 
\textit{vice versa} to calculate $Y_{treat}$.

\begin{figure*}[hbt!]
\makebox[\textwidth][c]{\includegraphics[scale=0.6]{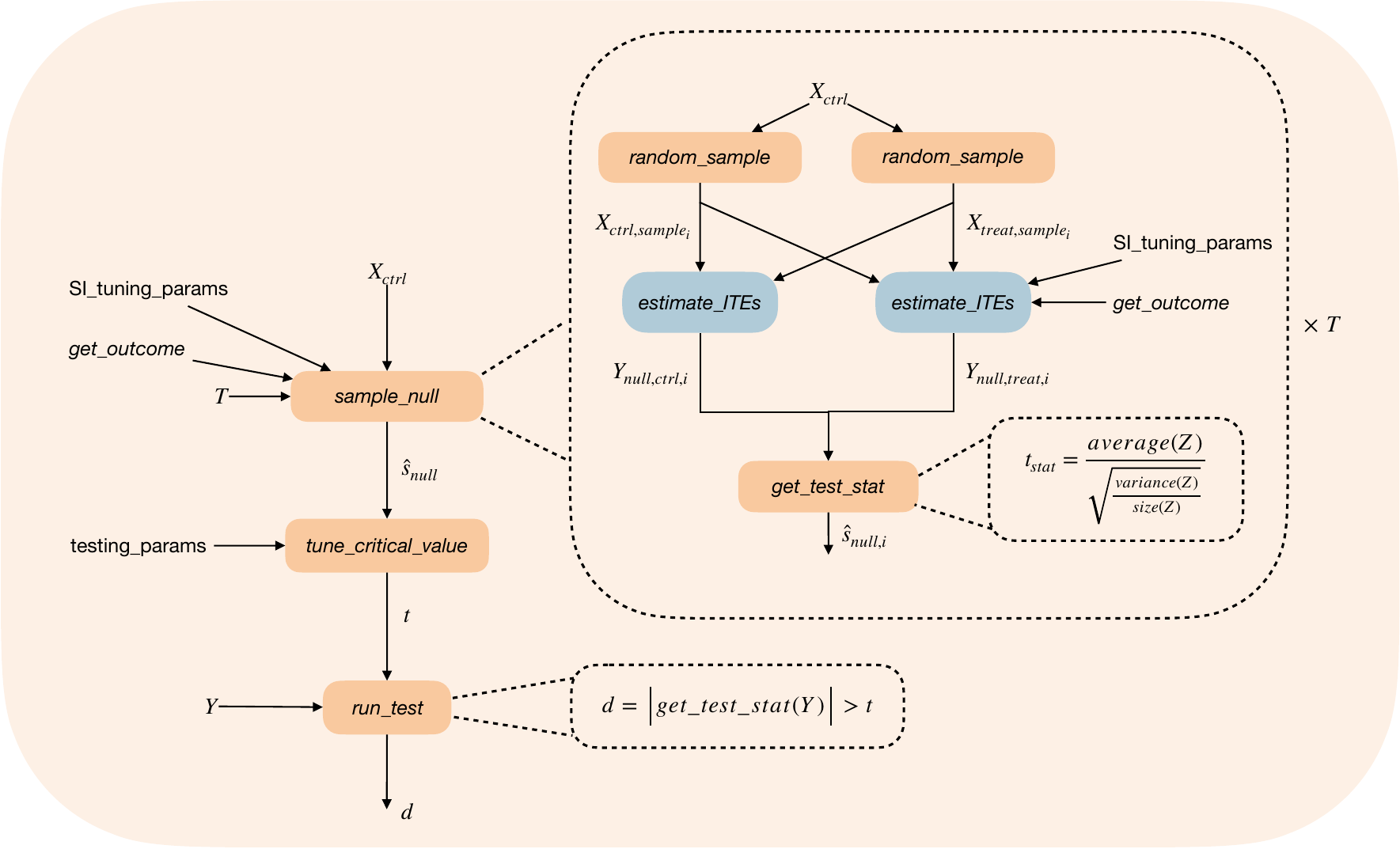}}
    \caption{Flowchart of \emph{run\_hypothesis\_test}.}
    \label{fig:run_hyp}
\end{figure*} 

By using the same donor data to estimate the ITE per patient in the same treatment arm, SI induces dependencies
among the ITEs, thereby violating the independence assumption under standard hypothesis testing. Therefore,
\emph{run\_hypothesis\_test} implements a new test, which uses the same test statistic used by the
one-sample \emph{t}-test but tunes the critical value of the test to control the significance level. The
flowchart of \emph{run\_hypothesis\_test} is shown in Fig.~\ref{fig:run_hyp}. First, it calls
\emph{sample\_null}, which estimates the variance of the null distribution of the test statistic using a
novel bootstrapping algorithm. The algorithm samples from the null distribution $T$ times by drawing bootstrap
samples from the original control arm data, $X_{ctrl}$, to construct comparable control and treatment groups,
$X_{ctrl,sample_i}$ and $X_{treat,sample_i}$, calculating their corresponding ITEs using
\emph{estimate\_ITEs} (with SI\_tuning\_params and \emph{get\_outcome}), and using the ITEs to calculate the
test statistic with \emph{get\_test\_stat}. Given samples of the test statistic from the null distribution
$\hat{s}_{null}$, \emph{run\_hypothesis\_test} calls \emph{tune\_critical\_value}, which tunes the
testing critical value to control the significance level using an algorithm similar to binary search, guided by a set of parameters defined under testing\_params, which include $\alpha_{target}$, the target significance level, $t_{lower}$, the lowerbound on the critical value search range, $t_{upper}$, the upperbound on the critical value search range, $t_{limit,exp}$, the amount to shrink or expand the critical value search range, $n_{s}$, the number of candidate critical values over which to search, and $\delta_{\alpha_{target}}$, the error tolerance for the significance level. \emph{run\_hypothesis\_test} then calls \emph{run\_test} to perform the test with the tuned critical value $t$ and the ITEs over the control and treatment groups $Y$, and returns the test decision $d$.

\section{Methodology} \label{sec:method}

In this section, we present TAD-SIE, a novel framework that efficiently identifies the sample size required to
reach a target operating point under a parallel-group trial in the absense of any prior knowledge of the ATE and 
variance, using a novel trend-adaptive design tailored to SECRETS. We present it for a two-arm superiority trial with 
an equal number of participants per arm, given that this is a simple but common setup for 
RCTs \cite{friedman2015fundamentals}. The 
flowchart of TAD-SIE is illustrated in Fig.~\ref{fig:framework}, which we walk through next.

TAD-SIE first calls \emph{conduct\_pilot\_study}, which runs an internal pilot study to obtain initial
estimates for these parameters under SECRETS. Specifically, it takes in as input $n_{pilot}$, the arm size for
the control/treament group under the pilot study, SI\_tuning\_params, a dictionary containing hyperparameters
for SI, and \emph{get\_outcome}, a function calculating the outcome metric from a patient's response trajectory,
both of which are used for estimating ITEs, as done in SECRETS, and $B$, a hyperparameter used for estimating the variance of the ITEs. \emph{conduct\_pilot\_study} outputs the initial RCT data for the control and treatment groups, $X_{ctrl,0}$ and $X_{treat,0}$, respectively, the arm size of the initial control and treatment groups, $n_{curr,0}$, along with initial estimates of the ATE and variance of the ITE, $\delta_0$ and $\sigma^2_0$, respectively. 

TAD-SIE then calls \emph{run\_sample\_size\_search} to iteratively refine the estimates to converge to an
accurate estimate of the target sample size needed for target power while stopping for futility to control the
significance level. To do this, it uses the RCT data collected from the pilot study, the initial estimates of
the ATE and variance, along with hyperparameters, which include $\alpha_{target}$, the target significance
level, $1-\beta_{target}$, the target power, $n_{max}$, the maximum arm size allowed by the trial,
\emph{step\_size\_scale\_factor}, the rate at which the sample size is increased, and
\emph{futility\_power\_boundary}, the threshold that determines if futility holds.  \emph{run\_sample\_size\_search} outputs the complete RCT data collected over the trial, $X_{ctrl,final}$ and $X_{treat,final}$, and \emph{futility\_flag}, which indicates whether the trial is futile. 

If the trial is futile, TAD-SIE fails to reject the null hypothesis and accepts it by convention  \cite{rosner2015fundamentals}. Otherwise, TAD-SIE runs SECRETS to perform hypothesis testing and return the test outcome $d$ given the complete RCT data, using the hyperparameters given by SI\_tuning\_params, testing\_params, and $T$, and function \emph{get\_outcome}. 

Next, we describe each step of TAD-SIE in more detail. 

\begin{figure}[hbt!]
\makebox[\textwidth][c]{\includegraphics[width=0.8\textwidth,angle=-90]{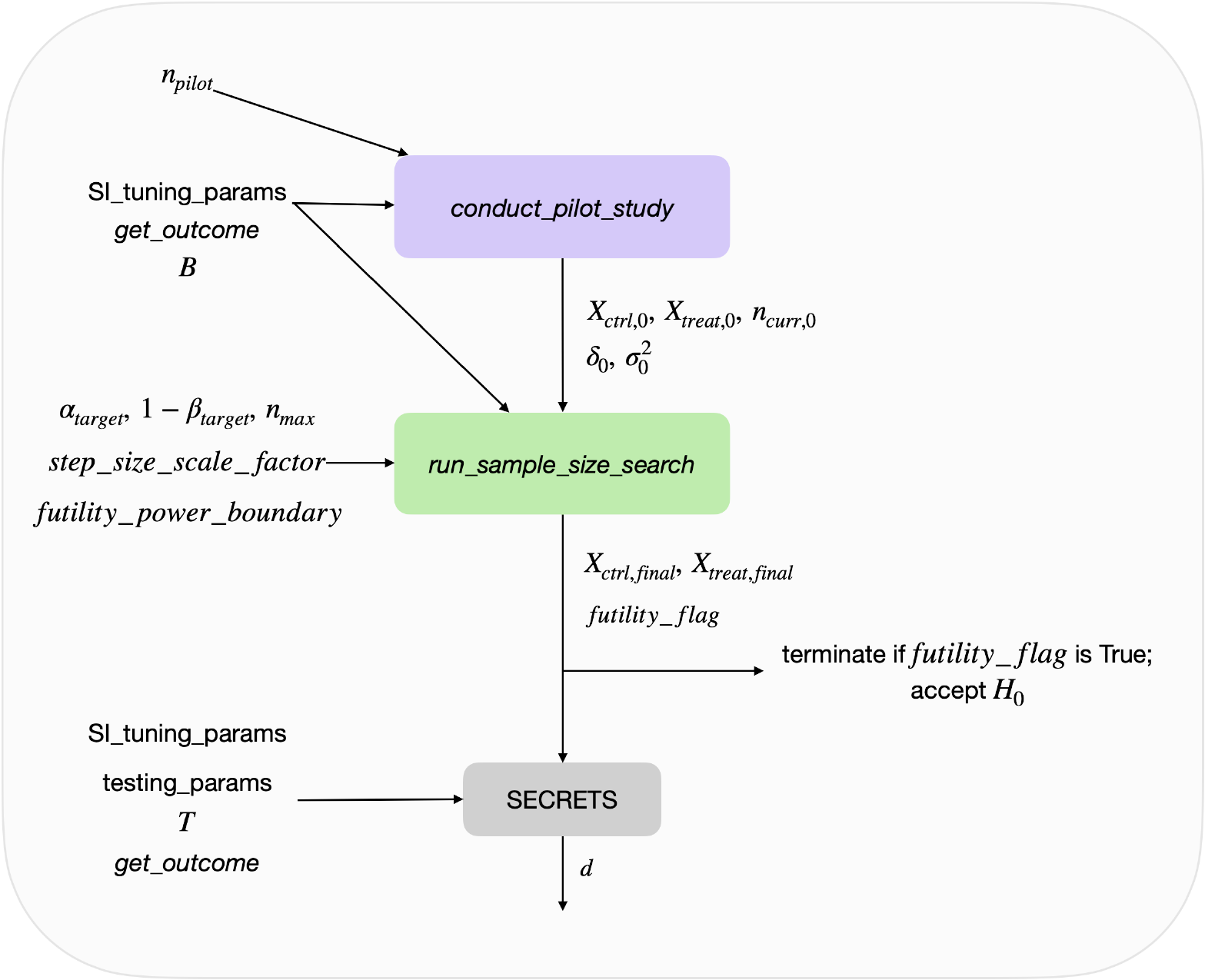}}
\caption{Flowchart of TAD-SIE.}
\label{fig:framework}
\end{figure}

\subsection{Initializing Search with a Pilot Study}

\emph{conduct\_pilot\_study} provides initial estimates on parameters defining the standard sample size
formula for a one-sample two-sided test (the testing scheme used by SECRETS), which is given by 
Eq.~(\ref{eq:sample_size}) \cite{rosner2015fundamentals}. Specifically, it estimates the moments, i.e., the ATE,
$\delta$, and variance of the ITE, $\sigma^2$, under SECRETS. The flowchart of \emph{conduct\_pilot\_study} is illustrated in Fig.~\ref{fig:pilot_study}, which we describe next.

\begin{equation}
    \label{eq:sample_size}
    n=\frac{\sigma^{2}(z_{1-\alpha/2}+z_{1-\beta})^{2}}{\delta^{2}}
\end{equation}

\begin{figure}[hbt!]
\makebox[\textwidth][c]{\includegraphics[width=0.7\textwidth,angle=-90]{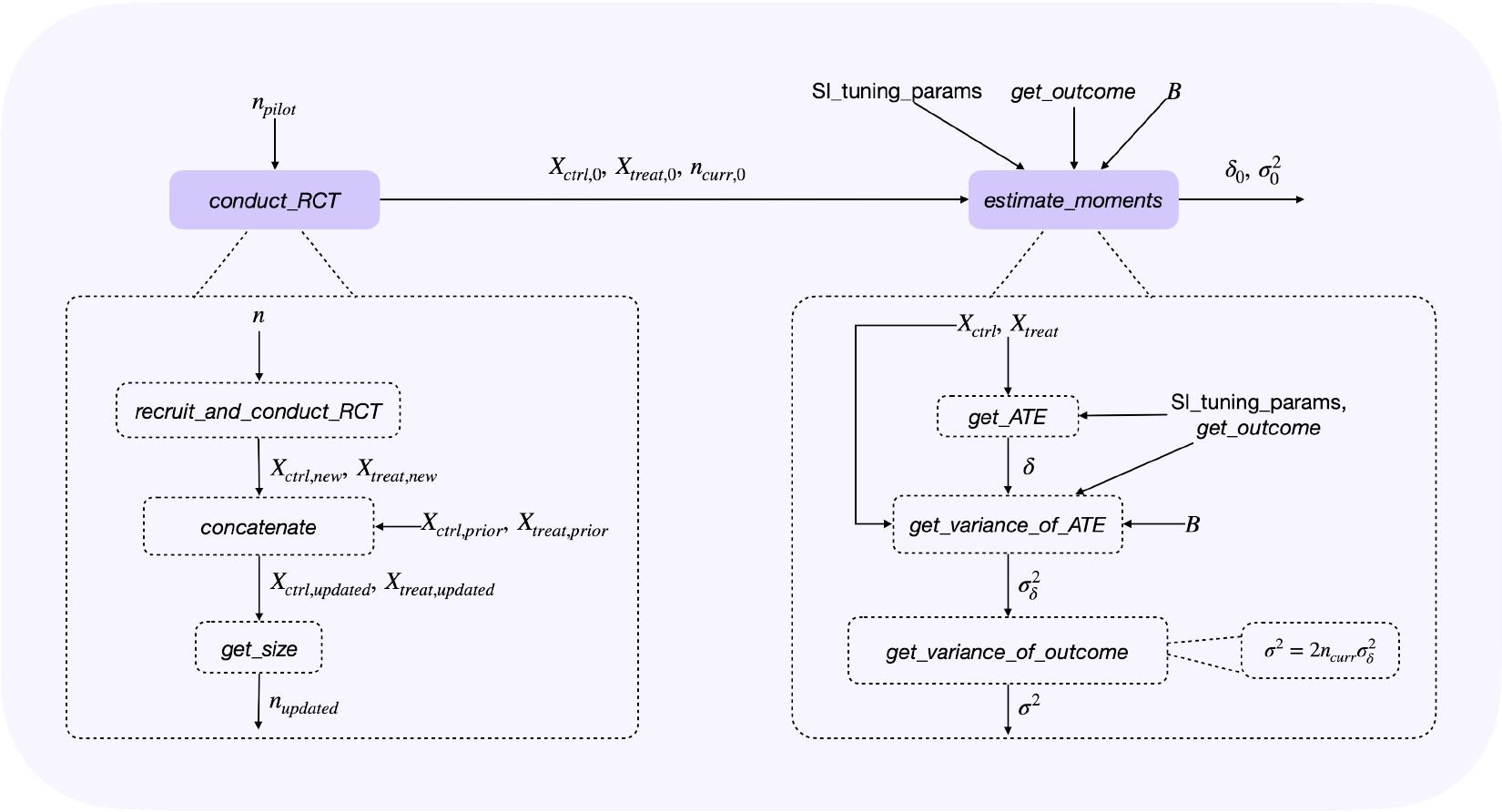}}
\caption{Flowchart of \emph{conduct\_pilot\_study}.}
\label{fig:pilot_study}
\end{figure}

To estimate $\delta$ and $\sigma^2$, \emph{conduct\_pilot\_study} runs an internal pilot study by collecting RCT data over a small arm size of $n_{pilot}$ with \emph{conduct\_RCT}. \emph{conduct\_RCT} calls \emph{recruit\_and\_conduct\_RCT} to run an RCT at the given arm size by recruiting subjects, randomizing them over control and treatment groups, and monitoring their responses up through the RCT endpoint. It then calls \emph{concatenate} to merge the newly collected RCT data over the control and treat arms, $X_{ctrl,new}$ and $X_{treat,new}$, respectively, with any prior collected data, $X_{ctrl,prior}$ and $X_{treat,prior}$ (there is no prior data for the pilot study). It also calculates the current arm size of the RCT with \emph{get\_size} and returns the updated data and arm size.  

Afterwards, \emph{conduct\_pilot\_study} estimates the moments under SECRETS from the RCT data collected
from the pilot study, $X_{ctrl,0}$ and $X_{treat,0}$, using \emph{estimate\_moments}, which  uses
SI\_tuning\_params, \emph{get\_outcome}, and hyperparameter $B$. \emph{estimate\_moments} first estimates
$\delta_0$ by calling \emph{get\_ATE}, illustrated in Fig.~\ref{fig:get_ate}. \emph{get\_ATE} first
calculates the ITEs, $Y_{ctrl}$ and $Y_{treat}$, for the control and treatment groups, respectively, using
\emph{estimate\_{ITEs}} with the collected RCT data, $X_{ctrl}$ and $X_{treat}$, hyperparameters defined by SI\_tuning\_params, and function \emph{get\_outcome}, and then averages the ITEs with \emph{get\_mean}.  

\begin{figure}[hbt!]
\makebox[\textwidth][c]{\includegraphics[width=0.35\textwidth,angle=-90]{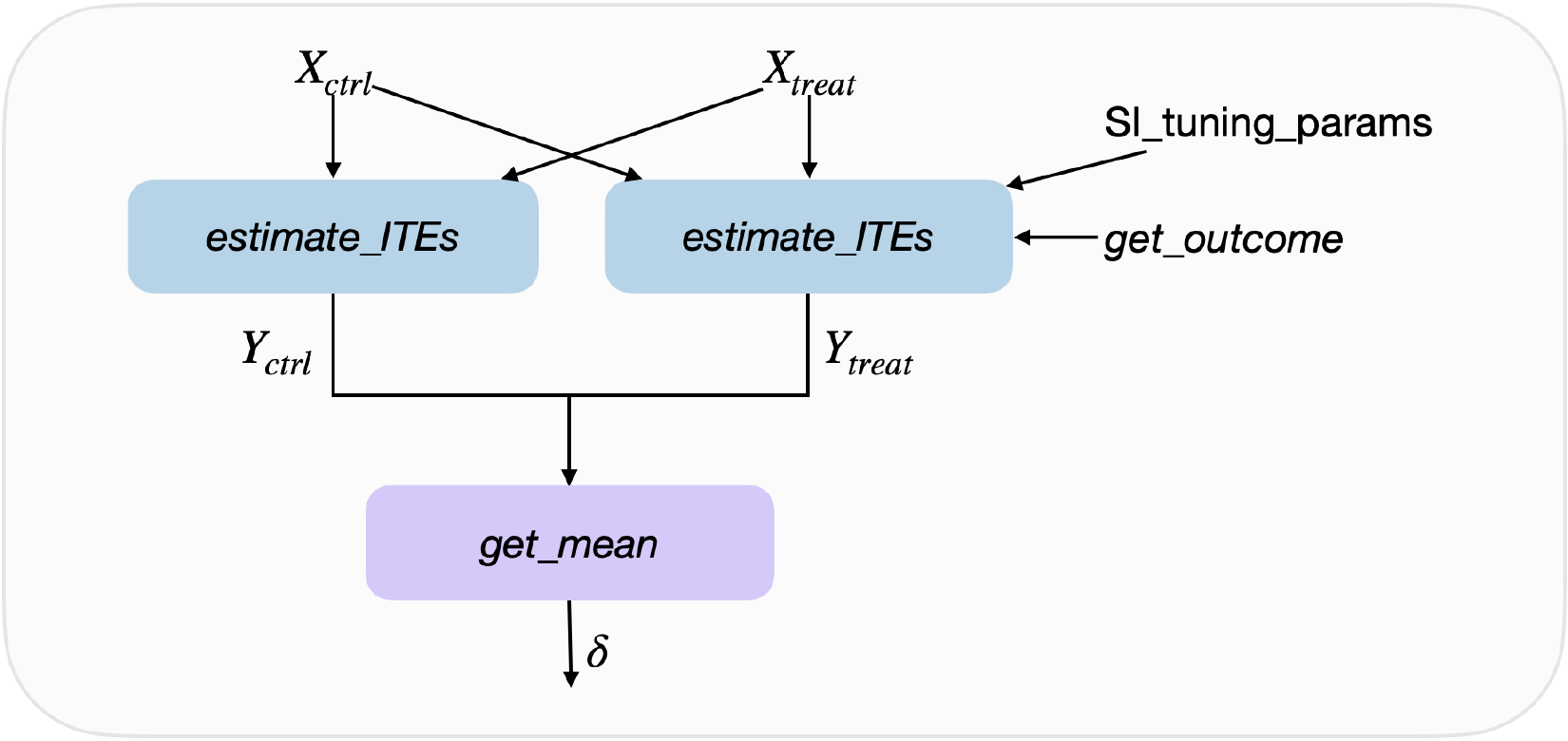}}
\caption{Flowchart of \emph{get\_ATE}.}
\label{fig:get_ate}
\end{figure}


Then, \emph{estimate\_moments} calculates $\sigma_0^2$ but uses a new procedure to do so because the sample
size formula given by Eq.~(\ref{eq:sample_size}) requires that the ITEs be independently and identically distributed (i.i.d.) according to a Gaussian, although the independence assumption is violated under SECRETS because SI induces dependencies among the ITEs. Given that the distribution of the ATE under SECRETS is approximated by a normal distribution, per the theorem on Dependency Neighborhoods based on Stein's method \cite{ross2011fundamentals}, \emph{estimate\_moments} first estimates $\sigma_{\delta}^{2}$, the variance of the ATE, and uses it to estimate $\sigma^{2}$, the variance of a set of hypothetical ITEs satisfying the i.i.d. assumption that would also yield the observed distribution of the ATE. 

To estimate $\sigma_{\delta}^{2}$, \emph{estimate\_moments} calls \emph{get\_variance\_of\_ATE},
illustrated in Fig.~\ref{fig:get_var_ate}. \emph{get\_variance\_of\_ATE} estimates the variance of the ATE
in a fashion similar to the \emph{sample\_null} procedure used in SECRETS for estimating the variance of the
ATE under the null distribution. Specifically, \emph{get\_variance\_of\_ATE} first samples from the
distribution of the ATE, where at each sampling iteration $b$, \emph{get\_variance\_of\_ATE} draws one
bootstrap sample from the RCT data over the control arm, $X_{ctrl,b}$, and one bootstrap sample from the RCT
data over the treatment arm, $X_{treat,b}$; this is different from \emph{sample\_null}, which samples from
the control arm only to produce samples from the null distribution. Afterwards,
\emph{get\_variance\_of\_ATE} calls \emph{get\_ATE} to calculate $\hat{\delta}_{b}$, the ATE associated
with the bootstrap samples. After sampling $B$ times from the distribution of the ATE, the variance of the ATE,
i.e., $\sigma_{\delta}^{2}$, is given by the variance over the $B$ samples. Afterwards,
\emph{estimate\_moments} calls \emph{get\_variance\_of\_outcome}, which calculates $\sigma^2$ according
to the relationship between variance of a mean and variance of underlying i.i.d. samples, given by 
Eq.~(\ref{eq:var_mean_sample}), where $\bar{X}$ is the mean over a set of $n$ i.i.d samples, $\{X_1,X_2,...,X_n\}$. Specifically, $\sigma^2$ is calculated by multiplying $\sigma_{\delta}^{2}$ by $2n_{curr}$ since the ITEs are pooled over the two arms, where $n_{curr}$ is the arm size of the control/treatment group.


\begin{figure}[hbt!]
\makebox[\textwidth][c]{\includegraphics[width=0.55\textwidth,angle=-90]{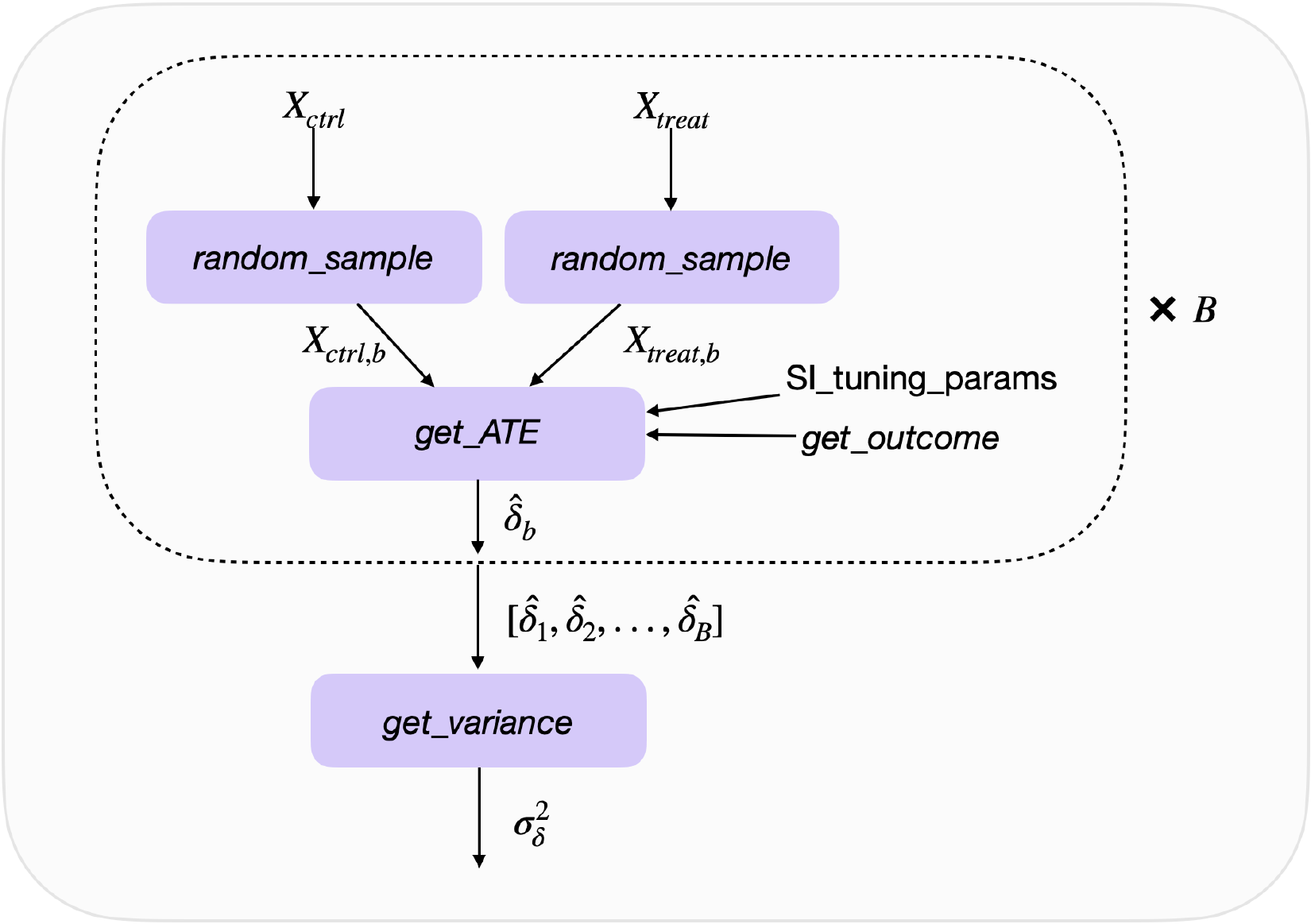}}
\caption{Flowchart of \emph{get\_variance\_of\_ATE}.}
\label{fig:get_var_ate}
\end{figure}

 \begin{align}
    \label{eq:var_mean_sample}
    \sigma_{\bar{X}}^{2} &\coloneqq \text{Var}(\bar{X}) \\
    &\coloneqq \text{Var}\biggr(\frac{\Sigma_{i=1}^{n} X_{i}}{n}\biggr) \nonumber \\
    &= \frac{\text{Var}(X_i)}{n} \nonumber \\ 
    &\coloneqq \frac{\sigma^2}{n} \nonumber
\end{align}

\clearpage

\subsection{Sample Size Estimation using a Trend-Adaptive Design}

The sample size formula works well when $\delta$ and $\sigma^2$ are estimated correctly. Hence, evaluations
based on estimates of $\delta$ and $\sigma^2$ from small samples (e.g., as with the pilot study) are expected to
extrapolate poorly. In addition, under TAD-SIE's procedure for estimating moments, the formula's extrapolation
capacity further diminishes because the variance can scale with sample size, as per Theorem \ref{thm:var_scale} (proof in \hyperref[sec:appendix]{Appendix}).

\begin{theorem} \label{thm:var_scale}

The variance calculated by estimate\_moments can scale with the sample size $n$, i.e., $\sigma^{2}=O(n)$. 

\end{theorem}

Therefore, TAD-SIE implements a trend-adaptive algorithm tailored to SECRETS, called \emph{run\_sample\_size\_search}, which iteratively accrues the RCT data, initialized from the pilot study, till a sufficient number of samples has been collected to reach target power (subject to resource constraints) or the trial terminates for futility. The 
flowchart of \emph{run\_sample\_size\_search} is illustrated in Fig.~\ref{fig:run_sample_size_search}, which we go through next.

\begin{figure}[hbt!]
\makebox[\textwidth][c]{\includegraphics[width=0.6\textwidth,angle=-90]{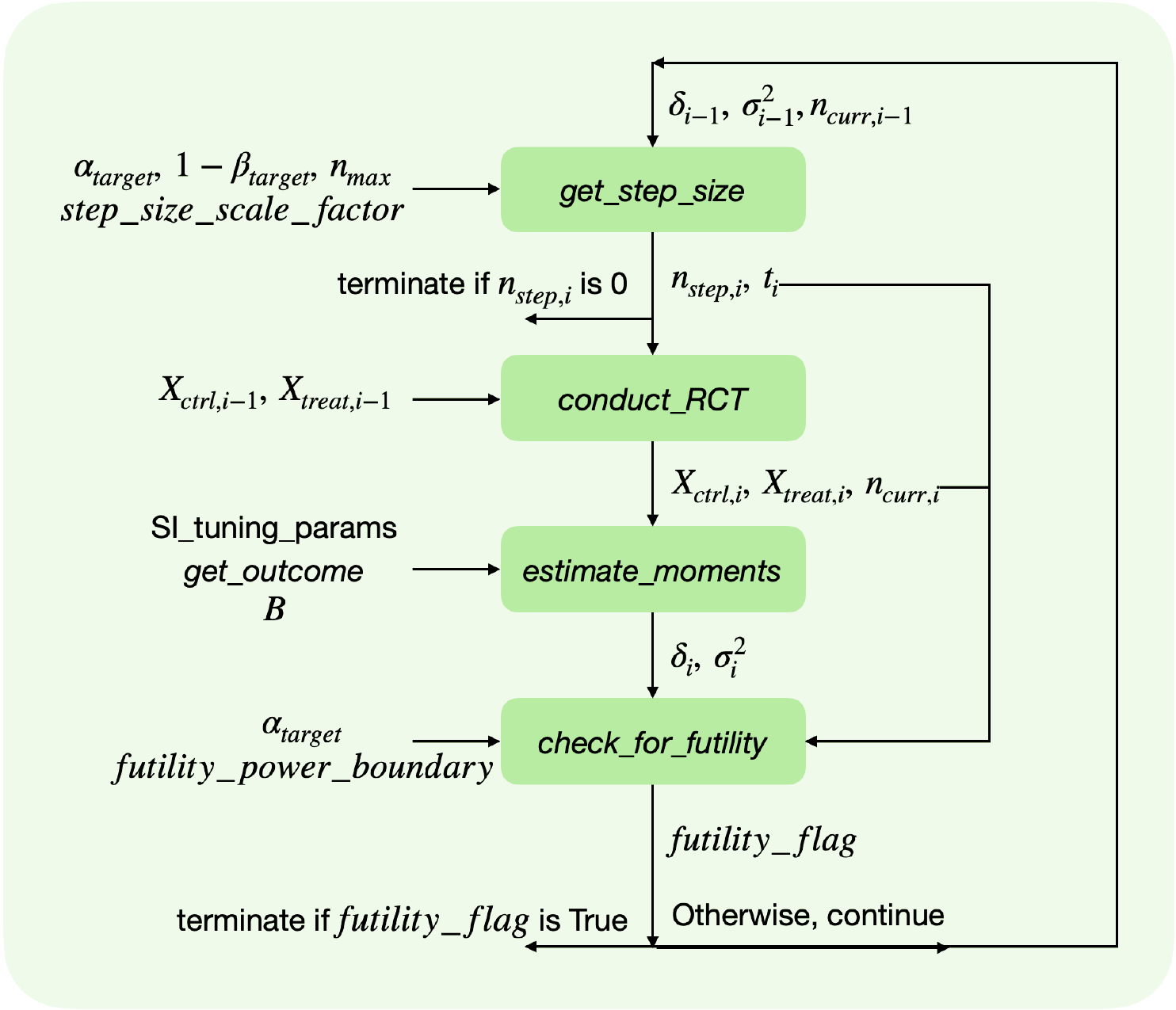}}
\caption{Flowchart of \emph{run\_sample\_size\_search}.}
\label{fig:run_sample_size_search}
\end{figure}   

At iteration $i$, \emph{run\_sample\_size\_search} first increases the current sample size according to
\emph{get\_step\_size}, which determines the step size $n_{step,i}$ and information fraction $t_i$ based on
quantities calculated from the prior iteration (pilot study if $i=0$), i.e., the moments $\delta_{i-1}$ and
$\sigma_{i-1}^2$ and arm size $n_{curr,i-1}$, and hyperparameters set by the trial designer, i.e.,
$\alpha_{target}$, $1-\beta_{target}$, $n_{max}$, and $step\_size\_scale\_factor$. The pseudocode for
\emph{get\_step\_size} is given in Alg.~\ref{alg:gss}. Specifically, \emph{get\_step\_size} first
calculates $n_{target}$, an estimate of the arm size needed for target performance, using the sample size
formula (Eq.~(\ref{eq:sample_size})) evaluated with target significance level $\alpha$, target power $1-\beta$,
and moments $\delta$ and $\sigma^{2}$; it divides by two since the total sample size is split between the
two arms (line \ref{agss:get_target}). Then, \emph{get\_step\_size} calculates the step size $n_{step}$ by
calculating the difference between the target arm size $n_{target}$ and current arm size $n_{curr}$ and scaling
it according to $step\_size\_scale\_factor$, which is less than 1 to guard against overpowering the study (line
\ref{agss:get_nstep_start}). The step size is constrained so that it is nonnegative (decreases are not allowed
since the data have already been collected) and does not exceed the maximum allowable step size set by the
maximum arm size $n_{max}$ (line \ref{agss:get_nstep_bounded}). In addition, \emph{get\_step\_size}
estimates the information fraction $t$ resulting from the sample size increase, which is used for futility stopping. $t$ is obtained by calculating the maximum possible step size given by $n_{step,max}$ (lines \ref{agss:get_nstep_max_start}-\ref{agss:get_nstep_max_bounded}) and taking the ratio between the updated arm size and the maximum possible value for the updated arm size, the latter quantity serving as an estimate of the final target arm size (line \ref{agss:get_t}).

\begin{algorithm*}[hbt!]
\caption{\emph{get\_step\_size}}
\label{alg:gss}
\begin{algorithmic}[1]

     \Statex \textbf{Input}:
     
     \Statex \hspace{\algorithmicindent} $\delta$: ATE under the alternative hypothesis $H_{1}$
     \Statex \hspace{\algorithmicindent} $\sigma^2$: variance of the ITEs 
     \Statex \hspace{\algorithmicindent} $n_{curr}$: current arm size in the RCT
     \Statex \hspace{\algorithmicindent} $\alpha \in [0,1]$: target significance level
     \Statex \hspace{\algorithmicindent} $1-\beta \in [0,1]$: target power
     
     \Statex \hspace{\algorithmicindent} $n_{max}$: maximum allowable arm size
     
     \Statex \hspace{\algorithmicindent} $step\_size\_scale\_factor \in (0,1]$: amount to scale the step size
     
     \Statex \textbf{Output}: 
     \Statex \hspace{\algorithmicindent} $n_{step}$: number of additional subjects to recruit per arm for the RCT 
      \Statex \hspace{\algorithmicindent} $t$: estimate of the information fraction under the new arm size
     
\State $n_{target}=\sigma^{2}(z_{1-\alpha/2}+z_{1-\beta})^{2}/(2\delta^{2})$ \label{agss:get_target}
     \State $n_{step}=(n_{target}-n_{curr}) \times step\_size\_scale\_factor$ \label{agss:get_nstep_start}
     \State $n_{step}=\min(\max(n_{step},0),n_{max}-n_{curr})$ \label{agss:get_nstep_bounded}
     
     \State $n_{step,max}=n_{target}-n_{curr}$ \label{agss:get_nstep_max_start}
     \State $n_{step,max}=\min(\max(n_{step,max},0),n_{max}-n_{curr})$ \label{agss:get_nstep_max_bounded}
     
	\State $t=(n_{curr}+n_{step})/(n_{curr}+n_{step,max})$ \label{agss:get_t}

     \State \textbf{return} $n_{step}, \ t$ \label{agss:get_nstep_end}
\end{algorithmic}
\end{algorithm*} 

If the step size is zero, \emph{run\_sample\_size\_search} terminates. Otherwise, it collects RCT data over $n_{step,i}$ additional subjects per arm using \emph{conduct\_RCT}, which also takes in the RCT data collected at the prior iteration, i.e., $X_{ctrl,i-1}$ and $X_{treat,i-1}$, to yield the accumulated RCT data over the control and treatment arms, $X_{ctrl,i}$ and $X_{treat,i}$, respectively, and the updated arm,  size $n_{curr,i}$. Afterwards, \emph{run\_sample\_size\_search} calls \emph{estimate\_moments} with the updated RCT data, $X_{ctrl,i}$ and $X_{treat,i}$, 
hyperparameters, i.e., SI\_tuning\_params and $B$, and function \emph{get\_outcome} to yield updated estimates of the moments, given by $\delta_i$ and $\sigma_i^2$. 

Finally, \emph{run\_sample\_size\_search} checks for futility using \emph{check\_for\_futility},
outlined in Alg.~\ref{alg:cff}, which uses quantities calculated in the current iteration of
\emph{run\_sample\_size\_search}, i.e., the moments $\delta_i$ and $\sigma_i^2$, the current arm size
$n_{curr,i}$, and information fraction $t_i$, along with hyperparameters $\alpha_{target}$ and
$futility\_power\_boundary$. Specifically, \emph{check\_for\_futility} assesses futility with CP, following
Stochastic Curtailment. First, \emph{check\_for\_futility} calculates the interim test statistic $z$ using
$\delta$, $\sigma^2$, and the current arm size, $n_{curr}$; the variance of the ATE is calculated using
Eq.~(\ref{eq:var_mean_sample}), where the total number of samples is $2n_{curr}$ since the ITEs are pooled over the two arms (line \ref{cff:z}). Then it evaluates CP at information fraction $t$, interim test statistic $z$, and target significance level $\alpha$ using the formula adapted for two-sided testing (line \ref{cff:cp}) \cite{rosner2015fundamentals} and marks the trial as futile if the estimate of CP is below $futility\_power\_boundary$ (line \ref{cff:futility_flag}). 

\begin{algorithm*}[hbt!]
\caption{\emph{check\_for\_futility}}
\label{alg:cff}
\begin{algorithmic}[1]

     \Statex \textbf{Input}:
     
     \Statex \hspace{\algorithmicindent} $\delta$: ATE under the alternative hypothesis $H_{1}$
     \Statex \hspace{\algorithmicindent} $\sigma^2$: variance of the ITEs 
     \Statex \hspace{\algorithmicindent} $n_{curr}$: current arm size in the RCT
     \Statex \hspace{\algorithmicindent} $t$: information fraction 
     \Statex \hspace{\algorithmicindent} $\alpha \in [0,1]$: target significance level
     \Statex \hspace{\algorithmicindent} $futility\_power\_boundary \in [0,1]$: the maximum power value for which a trial would be stopped for futility

     \Statex \textbf{Output}: 
     \Statex \hspace{\algorithmicindent} $futility\_flag$: indicator for whether futility holds  
\State $z=\delta/\sqrt{\sigma^2/(2n_{curr})}$ \label{cff:z}
	 \State $cp=\Phi\Big[z/\sqrt{t(1-t)}) - z_{1-\alpha/2}/\sqrt{1-t}\Big] + \Phi\Big[-z/\sqrt{t(1-t)}) - z_{1-\alpha/2}/\sqrt{1-t}\Big]$ \label{cff:cp}
	 \State $futility\_flag=cp \leq futility\_power\_boundary$ \label{cff:futility_flag}

     \State \textbf{return} $futility\_flag$ \label{cff:end}
\end{algorithmic}
\end{algorithm*} 

\emph{run\_sample\_size\_search} terminates if the trial is deemed futile; otherwise, it continues to the next iteration. The algorithm eventually terminates by futility or by reaching the maximum arm size $n_{max}$. Setting \emph{step\_size\_scale\_factor} determines how fast the algorithm terminates, with larger values resulting in fewer iterations at the cost of larger sample sizes since the sample size is increased at a higher rate. Given this tradeoff, the trial designer can implement a sample-efficient mode (TAD-SIE-SE) or a time-efficient mode (TAD-SIE-TE) by setting \emph{step\_size\_scale\_factor} to a low or high value, respectively.

\subsection{Hypothesis Testing}

After determining the sample size, TAD-SIE performs hypothesis testing. If \emph{run\_sample\_size\_search} marks the trial as futile, TAD-SIE fails to reject the null hypothesis, in which case it accepts it by convention \cite{rosner2015fundamentals}. Otherwise, TAD-SIE performs hypothesis testing with SECRETS using the final RCT data given by $X_{ctrl,final}$ and $X_{treat,final}$, along with hyperparameters contained in SI\_tuning\_params and testing\_params, hyperparameter $T$, and function \emph{get\_outcome}, and returns the test outcome $d$.

\section{Performance Evaluation} \label{sec:eval}

In this section, we describe performance metrics and the dataset used to evaluate TAD-SIE. We describe the baseline algorithms against which we compare and ablation studies we perform to establish the premises underlying TAD-SIE. Finally, we provide details on the implementation of the algorithms and experiments.  

\subsection{Performance Metrics}

We set target power, $1-\beta_{target}$, to 80\% or 90\% and target significance level, $\alpha_{target}$, to
5\%, following typical target operating points \cite{friedman2015fundamentals}. We measure power and significance level obtained by TAD-SIE and baseline algorithms following the approach from \cite{blackston2019comparison}, which simulates many trials under the alternative and 
null settings and calculates the percentage of trials where the test procedure returns a reject, respectively. Specifically, we simulate a trial under the 
alternative setting by constructing new control and treatment arms with subjects sampled with replacement from the original RCT's control and treatment arms, respectively. Similarly, we simulate the null setting by constructing both the control and treatment arms 
with subjects sampled with replacement from the original RCT's control arm. 
  
For TAD-SIE, we also report the final arm size and number of iterations that a trial takes in order to characterize TAD-SIE's efficiency. We show the distribution of these quantities across trials using a box plot with the whisker length set at 1.5 times the interquartile range and outliers removed.

\subsection{Dataset} \label{sec:datasets}

We evaluate the framework on a real-world clinical Phase-3 parallel-group RCT and demonstrate it for a two-arm superiority trial, a design typically adopted in clinical RCTs \cite{friedman2015fundamentals}. We obtained the dataset for a sample trial, e.g.,  CHAMP (NCT01581281), \cite{powers2017trial,champ_stats}, from the National Institute of Neurologic Disease and Stroke (NINDS) \cite{ninds_archive}. The CHAMP study \cite{powers2017trial} conducted an RCT to compare the effect of different medications (amitriptyline and topiramate) on mitigating headaches. Following the setup from \cite{lala2023secrets}, we construct a dataset corresponding to a two-arm superiority trial. Specifically, we set the control arm to be the group exposed to amitriptyline and the treatment arm to be the group exposed to topiramate and define the ATE to be the difference between their average outcomes, where the outcome is the change in the score on the Pediatric Migraine Disability Assessment Scale between the 24-week endpoint and the baseline visit. This setup yields a dataset of 204 subjects, with 106 in the amitriptyline group and 98 in the topiramate group, where the ATE is -3.17 units, somewhat comparable to the -4.3 units reported in the study \cite{champ_stats}.

\subsection{Baselines}

We compare TAD-SIE against two baseline algorithms. Both algorithms implement parallel-group RCTs following a two-arm superiority setup and therefore use the two-sample \emph{t}-test for independent samples with unequal variances for hypothesis testing \cite{rosner2015fundamentals}. The approaches differ in how they determine the final sample size. 

The Fixed Sample Design baseline is a standard approach for study planning that calculates the sample size
required for target power and target significance level using Eq.~(\ref{eq:sample_size_two_sample}), where the
ATE $\delta$ and variances, $\sigma_{ctrl}^2$ and $\sigma_{treat}^2$, are pre-specified or estimated from a
prior study \cite{friedman2015fundamentals,mehta2011adaptive}. Since domain knowledge may not be available to
appropriately pre-specify these parameters and since prior comparable studies do not exist for new interventions, the baseline implements a small internal pilot study to estimate these parameters \cite{hulley2013designing}. The baseline then conducts an RCT according to the calculated sample size, which is capped at a maximum arm size set by trial constraints \cite{mehta2011adaptive}.  

Since the Fixed Sample Design baseline performs poorly in reaching a target operating point, given that it relies 
on noisy estimates of the ATE and variances, we implement an adaptive design that can increase the initial
sample size calculated by the Fixed Sample Design strategy in order to increase power. We adopt a TAD based on
CP and specifically implement the algorithm from \cite{chen2004increasing}, given its simplicity, and use the sample size formula to calculate the amount by which to increase the sample size for simplicity. We refer to this baseline as Standard-TAD.     

\subsection{Ablations}

We perform ablations on TAD-SIE in order to understand the contribution of key steps. First, we ablate the 
\emph{estimate\_moments} procedure by swapping our variance estimation strategy with a naive one to show
that our approach can reach the target operating point and do so more efficiently. Then, we run ablations on our
trend-adaptive algorithm by swapping it with a standard TAD \cite{chen2004increasing} to show that an approach
implementing a rule for sample size increases, based on control over significance level, will fail to reach the
target operating point since increases are rare. Finally, we modify TAD-SIE so that it performs sample size
estimation based on standard hypothesis testing instead of SECRETS to show that a TAD designed for a powerful test is necessary for reaching the target operating point.

For all ablations, we sweep over any relevant hyperparameters to maximize the ablated algorithm's performance. We also report results for when target power is 80\% and target significance level is 5\% since the effect of each ablation is independent of the target operating point.

\subsection{Implementation Details}

First, we describe the hyperparameters and other inputs used by TAD-SIE and the baseline algorithms. Then, we 
describe the hyperparameters used for evaluation. Finally, we describe the computing setup used to perform experiments. 

Hyperparameters for TAD-SIE are reported in Table~\ref{tab:alg_hyperparams}. While most hyperparameters can be 
determined from prior work, \emph{step\_size\_scale\_factor} is a new hyperparameter introduced by TAD-SIE;
hence, we sweep over values over the domain of the hyperparameter in increments of 0.1 to characterize its effect on performance. We simultaneously sweep over values for \emph{futility\_power\_boundary} over a sufficient range in increments of 1\% starting from 1\% since there is no standard or recommended value \cite{snapinn2006assessment}; we also include 0\% since this corresponds to no futility stopping. In addition, we set \emph{get\_outcome} to be the function 
that calculates the outcome metric defined for the dataset in Sec.~\ref{sec:datasets}. 

The baseline algorithms use the same values used by TAD-SIE for $n_{pilot}$, $\alpha_{target}$, $1-\beta_{target}$, and $n_{max}$, and calculate the same outcome metric. For Standard-TAD, we set the number of interim analyses to 1 since this is common in practice \cite{mehta2011adaptive} and perform interim analysis at the initial planned sample size by setting $t=0.99$ (CP is undefined at $t=1$) since this is ideal for assessing whether the sample size can be increased and the amount by which it needs to be increased \cite{mehta2011adaptive}. 

\begin{table*}[h!]
\centering
\caption{Hyperparameters used for TAD-SIE. The ``Reference" column lists references that support the choice of the hyperparameter value.}
\label{tab:alg_hyperparams} 
\resizebox{\columnwidth}{!}{%
\begin{tabular}{|M{0.33\linewidth} M{0.33\linewidth} M{0.33\linewidth}|}
\hline
\textbf{Hyperparameter} & \textbf{Value} & \textbf{Reference} \\ \hline 
$n_{pilot}$ & 30 & \cite{teare2014sample,birkett1994internal} \\ \hline 
SI\_tuning\_params & \{ $r_{train,val}=7/3$ \} & \cite{lala2023secrets} \\ \hline 
$B$ & 100 & \cite{lala2023secrets} \\ \hline
$\alpha_{target}$ & 5\% & \cite{friedman2015fundamentals} \\ \hline 
$1-\beta_{target}$ & 80\% / 90\%  & \cite{friedman2015fundamentals} \\ \hline 
$n_{max}$ & 1500 & \cite{fda_phases} \\ \hline 
\emph{step\_size\_scale\_factor} &  $(0,1]$ & n/a \\ \hline 
\emph{futility\_power\_boundary} & $[0,20\%]$ & n/a \\ \hline
testing\_params & \{ $\alpha_{target}=5\%$, $t_{lower}=3$, $t_{upper}=5$, $t_{limit,exp}=2$, $n_{s}=10$, $\delta_{\alpha_{target}}=$ $1e$-$3$ \} & \cite{lala2023secrets} \\ \hline 
$T$ & 100 & \cite{lala2023secrets}  \\ \hline 
\end{tabular}%
}
\end{table*}
  
For evaluation, we set the number of trials to 100 since this was sufficient for powers and significance levels to stabilize. For implementation ease, we pre-computed results across arm sizes sampled between the pilot study size to the maximum arm size and then projected interim and final sample sizes to this set. We sampled in increments of 25 since this was sufficient for consecutively sampled arm sizes to have comparable means and variances.

We implemented the framework and experiments with Python using standard numerical packages and conducted experiments using 28-32 CPU cores, 2-4GB of memory per CPU, and Intel processors (e.g., 2.4 GHz Skylake and 2.6 GHz Intel Skylake).

\section{Experimental Results} \label{sec:results}

This section presents results from our experiments. First, we characterize the effect of key hyperparameters underlying TAD-SIE; given these results, we determine adequate hyperparameter settings for the sample-efficient and time-efficient modes of TAD-SIE. Then, we demonstrate the effectiveness of each mode of TAD-SIE by comparing them against the baselines. Finally, we analyze results from the ablation studies to understand the core components enabling TAD-SIE.

\subsection{Hyperparameter Selection} \label{sec:hyperparam}

In this section, we present results demonstrating the effect of hyperparameters that govern dynamics under TAD-SIE. First, we assess the effect of \emph{step\_size\_scale\_factor} and \emph{futility\_power\_boundary} on power and significance level to identify a range of hyperparameter values that are feasible, i.e., can reach the target operating point. Given the feasible range, we analyze the effect of \emph{step\_size\_scale\_factor} on the sample and time efficiencies of solutions under TAD-SIE in order to suggest hyperparameter settings for the sample-efficient and time-efficient modes. 

The effect of \emph{step\_size\_scale\_factor} and \emph{futility\_power\_boundary} on TAD-SIE's performance,
when target significance level is 5\% and target power is 80\%, is shown in Fig.~\ref{fig:feasible_set}. As expected,
setting \emph{futility\_power\_boundary} to 0\% fails to control the significance level (significance levels are
around 10\%); increasing it reduces the significance level and power across all values for
\emph{step\_size\_scale\_factor}, although feasible solutions are less likely to exist when
\emph{step\_size\_scale\_factor} is too large, i.e., values larger than 0.7. For example, when
\emph{step\_size\_scale\_factor} is 0.9, setting $\emph{futility\_power\_boundary}$ to just 1\% already reduces
the significance level to 5\% and power to 78\% and when \emph{step\_size\_scale\_factor} is 1, setting
$\emph{futility\_power\_boundary}$ to 1\% reduces the significance level to 4\% and power to 67\%. In contrast, when \emph{step\_size\_scale\_factor} is 0.5, setting \emph{futility\_power\_boundary} between 3\% and 7\% obtains 5\% significance level 
and at least 80\% power, with similar results when \emph{step\_size\_scale\_factor} is 0.1 and 0.3. Similar trends hold when target power is 90\%, as shown in Fig.~\ref{fig:feasible_set_0.1}.  

\begin{figure}[hbt!]
\makebox[\textwidth][c]{\includegraphics[width=1.1\textwidth]{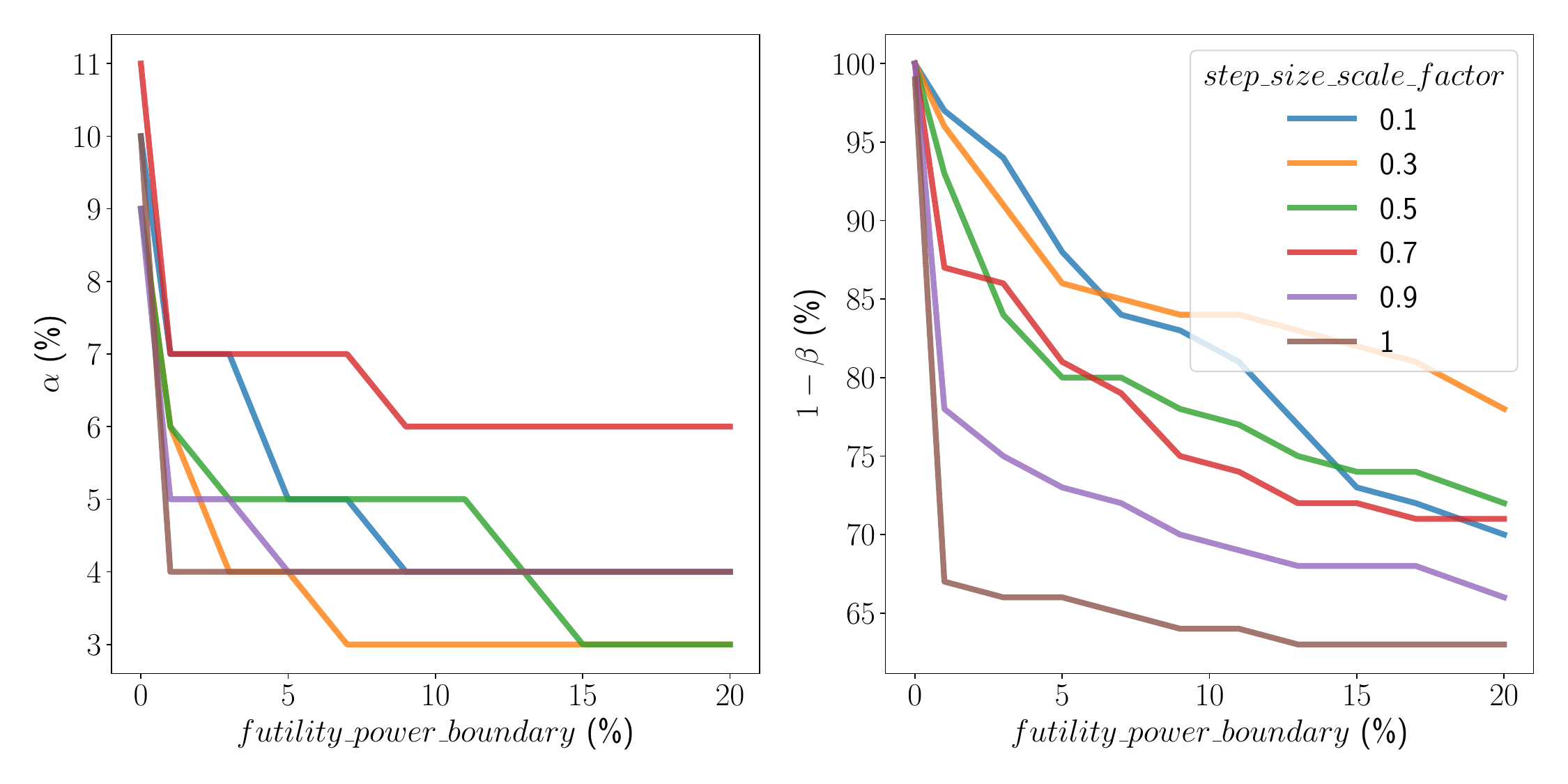}}
    \caption{Effect of \emph{futility\_power\_boundary} and \emph{step\_size\_scale\_factor} on the significance level and power under TAD-SIE with $1-\beta_{target}=80\%$ and $\alpha_{target}=5\%$.}
    \label{fig:feasible_set}
\end{figure}

\begin{figure}[hbt!]
\makebox[\textwidth][c]{\includegraphics[width=1.1\textwidth]{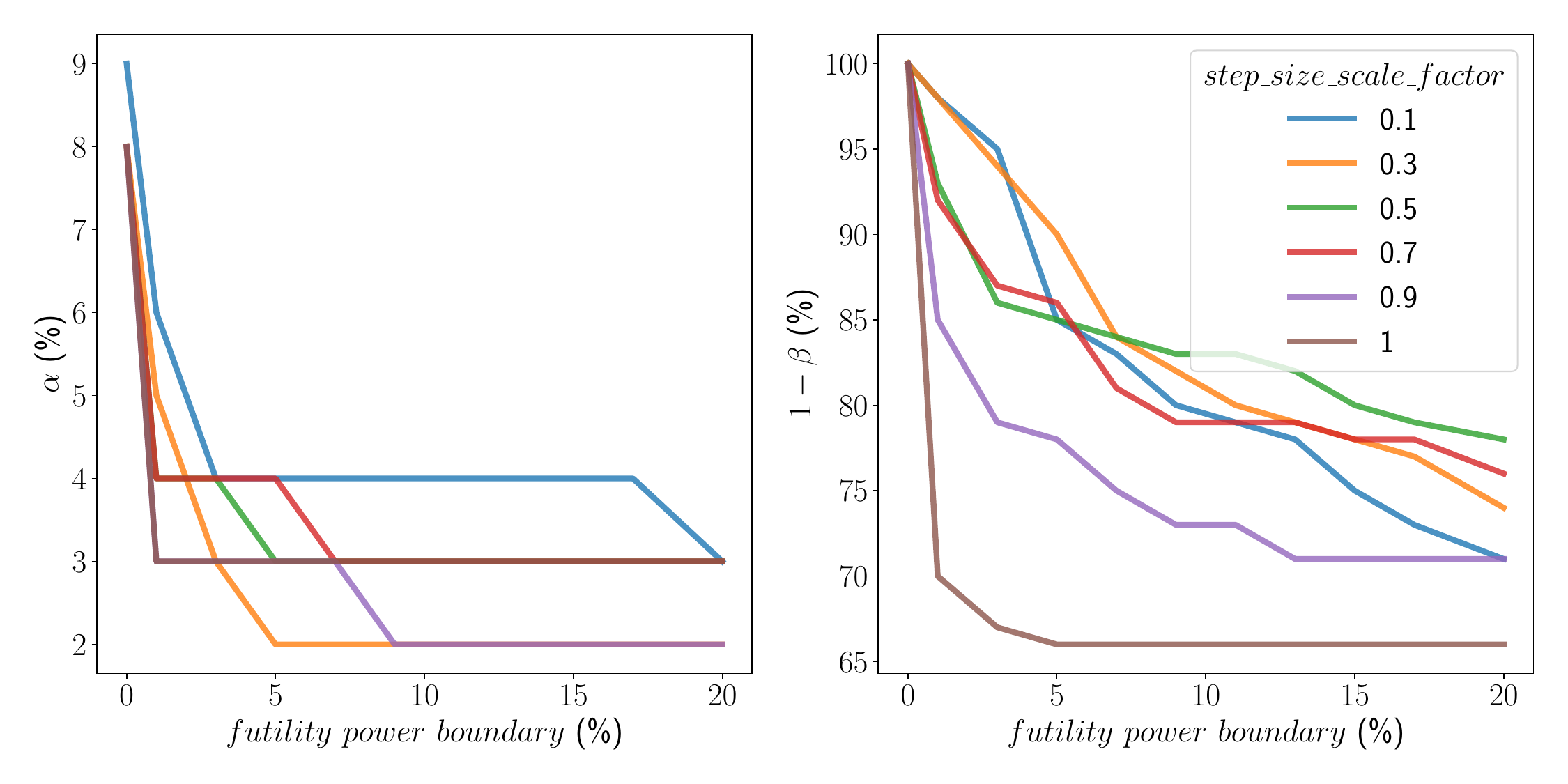}}
    \caption{Effect of \emph{futility\_power\_boundary} and \emph{step\_size\_scale\_factor} on the significance level and power under TAD-SIE with $1-\beta_{target}=90\%$ and $\alpha_{target}=5\%$.}
    \label{fig:feasible_set_0.1}
\end{figure}

Feasible solutions are unlikely to exist when setting \emph{step\_size\_scale\_factor} to large values because doing so increases the chance for futility stopping when $H_1$ holds, which makes it difficult to reach target power while 
controling the significance level. For example, when \emph{futility\_power\_boundary} is 1\%, 7\% and 22\% of trials terminate by futility when $H_1$ holds under a \emph{step\_size\_scale\_factor} of 0.5 and 0.9, respectively. Setting \emph{step\_size\_scale\_factor} to a large value increases the chance for futility stopping by increasing the information fraction $t$, given that $t$ is more likely to be determined by \emph{step\_size\_scale\_factor} at initial interim analyses when the current arm size is small (per line \ref{agss:get_t} in Alg.~\ref{alg:gss}). A higher information fraction shrinks CP over a large range of test statistic values, as shown in Fig.~\ref{fig:t_vs_cp}, thereby triggering futility stopping. Given the relationship between \emph{step\_size\_scale\_factor} and futility stopping, we only consider solutions where \emph{step\_size\_scale\_factor} is not too large. Setting \emph{step\_size\_scale\_factor} below 0.7 works well in yielding a set of feasible solutions across typical target operating points (i.e., 80\% or 90\% power and 5\% significance level).

\begin{figure}[hbt!]
\makebox[\textwidth][c]{\includegraphics[width=0.6\textwidth]{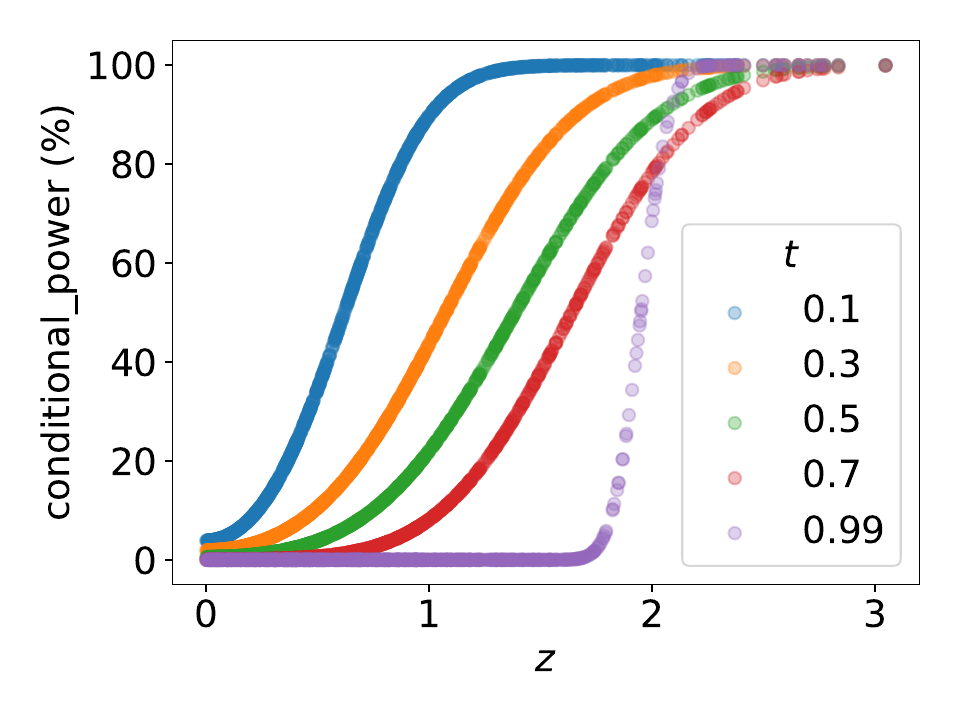}}
    \caption{Conditional power under two-sided testing as a function of $z$, the interim test statistic, and $t$, the information fraction (formula specified in line \ref{cff:cp} of Alg.~\ref{alg:cff}). Function plotted over nonnegative values of $z$ since the function is symmetric in $z$.}
    \label{fig:t_vs_cp}
\end{figure}

By adopting smaller values for \emph{step\_size\_scale\_factor}, TAD-SIE is more robust to futility stopping
when $H_1$ holds and futility stopping rates remain comparable when $H_0$ holds since the ATE is small.
Therefore, values for \emph{futility\_power\_boundary} generally transfer well, i.e., they yield operating
points in a similar range across other values for \emph{step\_size\_scale\_factor}. This phenomenon is shown in
Fig.~\ref{fig:futility_boundary} when the target significance level is 5\% and target power is 80\%. For
example, when \emph{futility\_power\_boundary} is 1\%, setting \emph{step\_size\_scale\_factor} to any value
between 0.1 and 0.6 yields significance levels in a higher range (5-7\%) and powers in a high range (88-97\%)
and when \emph{futility\_power\_boundary} is 20\%, setting \emph{step\_size\_scale\_factor} to any value between
0.1 and 0.6 yields significance levels in a lower range (3-4\%) and powers in a low range (70-78\%). Among
hyperparameter settings that meet the target operating point, indicated by the region enclosed by dashed lines in Fig.~\ref{fig:futility_boundary}, setting \emph{futility\_power\_boundary} to values near 10\% work well across values for \emph{step\_size\_scale\_factor} between 0.1 and 0.6; therefore, we set it to 11\%. By applying a similar process to TAD-SIE when target power is 90\% and target significance level is 5\%, we find that setting \emph{futility\_power\_boundary} to values near 1\% work well, as shown in Fig.~\ref{fig:futility_boundary_0.1} (the boundary needs to be reduced to increase power); therefore, we set it to 1\% in this case.

\begin{figure}[hbt!]
	\makebox[\textwidth][c]{\includegraphics[width=0.5\textwidth,angle=-90]{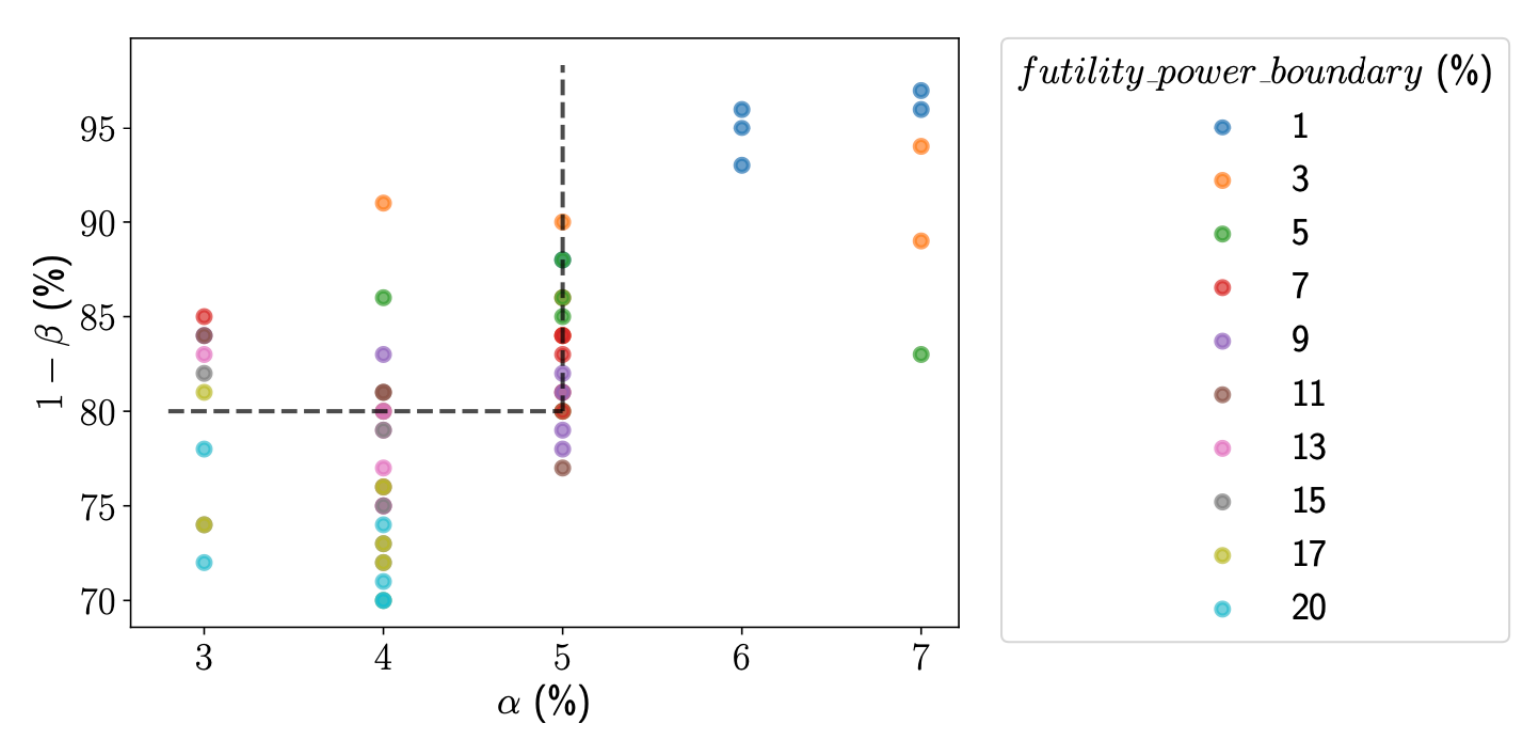}}
    \caption{Significance level and power under different values of \emph{futility\_power\_boundary} and
\emph{step\_size\_scale\_factor} for TAD-SIE with $1-\beta_{target}=80\%$ and $\alpha_{target}=5\%$. Each point
corresponds to a different value of \emph{step\_size\_scale\_factor} that has been sampled over the following
set of values: $\{0.1,0.2,0.3,0.4,0.5,0.6\}$.}
    \label{fig:futility_boundary}
\end{figure}

\begin{figure}[hbt!]
	\makebox[\textwidth][c]{\includegraphics[width=0.5\textwidth,angle=-90]{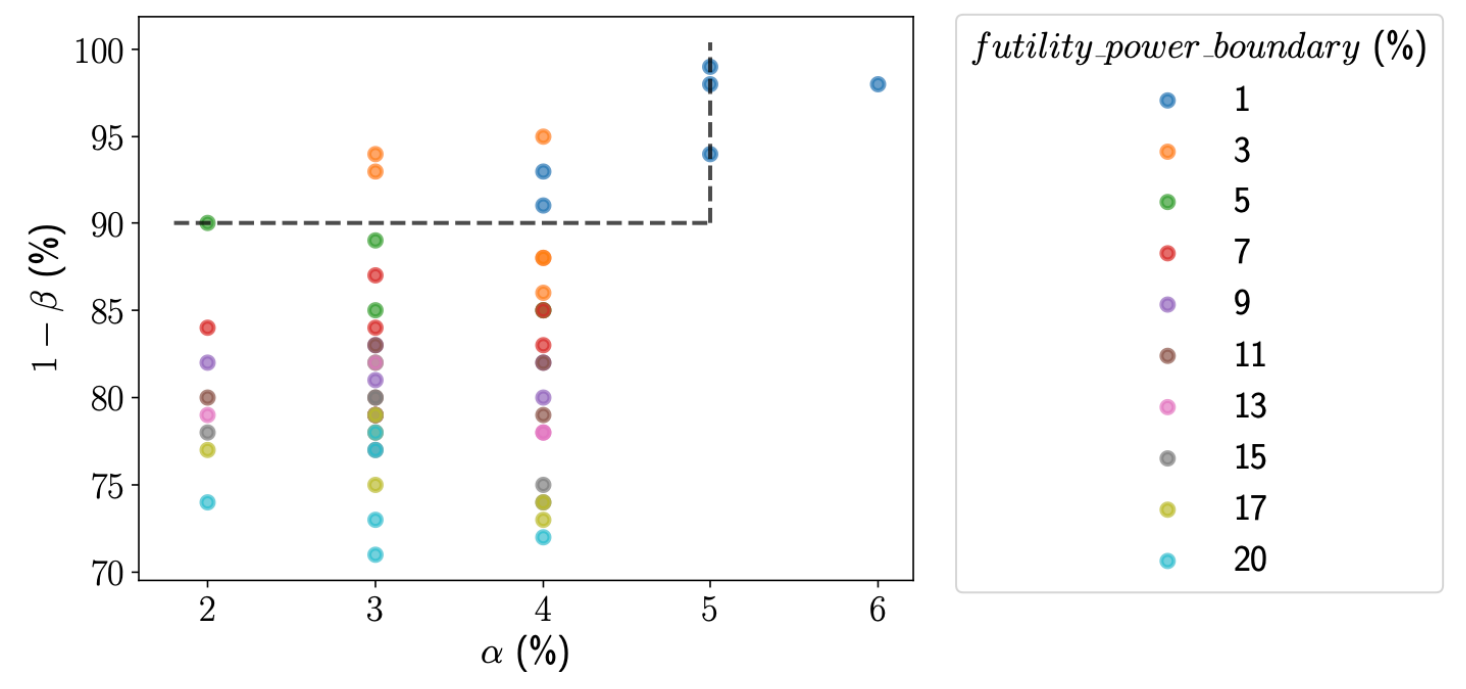}}
    \caption{Significance level and power under different values of \emph{futility\_power\_boundary} and \emph{step\_size\_scale\_factor} for TAD-SIE with $1-\beta_{target}=90\%$ and $\alpha_{target}=5\%$. Each point corresponds to a different value of \emph{step\_size\_scale\_factor} that has been sampled over the following set of values: $\{0.1,0.2,0.3,0.4,0.5,0.6\}$.}
    \label{fig:futility_boundary_0.1}
\end{figure}

Given an effective value for \emph{futility\_power\_boundary}, we vary \emph{step\_size\_scale\_factor} to yield
solutions that reach the target operating point but are efficient in different respects, i.e., the final arm size and
number of iterations incurred under TAD-SIE. We demonstrate this tradeoff when target power is 80\% and target
significance level is 5\%, as shown in Fig.~\ref{fig:sample_size_scale_factor_0.2}. For example, when $H_1$ holds and
\emph{step\_size\_scale\_factor} is 0.1, the median final arm size is 375 and the median number of iterations is 5 and
when \emph{step\_size\_scale\_factor} is 0.6, the median final arm size is 688 and the median number of iterations is
2. This tradeoff occurs because increasing \emph{step\_size\_scale\_factor} increases initial step sizes, which can
overshoot and cause the algorithm to terminate earlier since the step size goes to zero in a later iteration; in
addition, increasing \emph{step\_size\_scale\_factor} increases the chance for futility stopping through its effect on
information fraction, which can also cause trials to terminate earlier. When $H_0$ holds, similar trends manifest;
however, since the ATE is small, this triggers futility stopping, which generally keeps the number of iterations and
arm size lower than those observed under $H_1$. For example, when \emph{step\_size\_scale\_factor} is 0.1, the median
final arm size is 300 and the median number of iterations is 4 and when \emph{step\_size\_scale\_factor} is 0.6, the
median final arm size is 538 and the median number of iterations is 2. Similar trends hold when the target power is
90\% and target significance level is 5\%, as shown in Fig.~\ref{fig:sample_size_scale_factor_0.1}, although the
sample size generally tops out under $H_0$, even when \emph{step\_size\_scale\_factor} is low, since target power is higher. Based on these trends, we set \emph{step\_size\_scale\_factor} to 0.1 to implement a sample-efficient version of TAD-SIE, i.e., TAD-SIE-SE, and \emph{step\_size\_scale\_factor} to 0.6 to implement a time-efficient version of TAD-SIE, i.e., TAD-SIE-TE. 

\begin{figure}[hbt!]
	\makebox[\textwidth][c]{\includegraphics[width=1\textwidth]{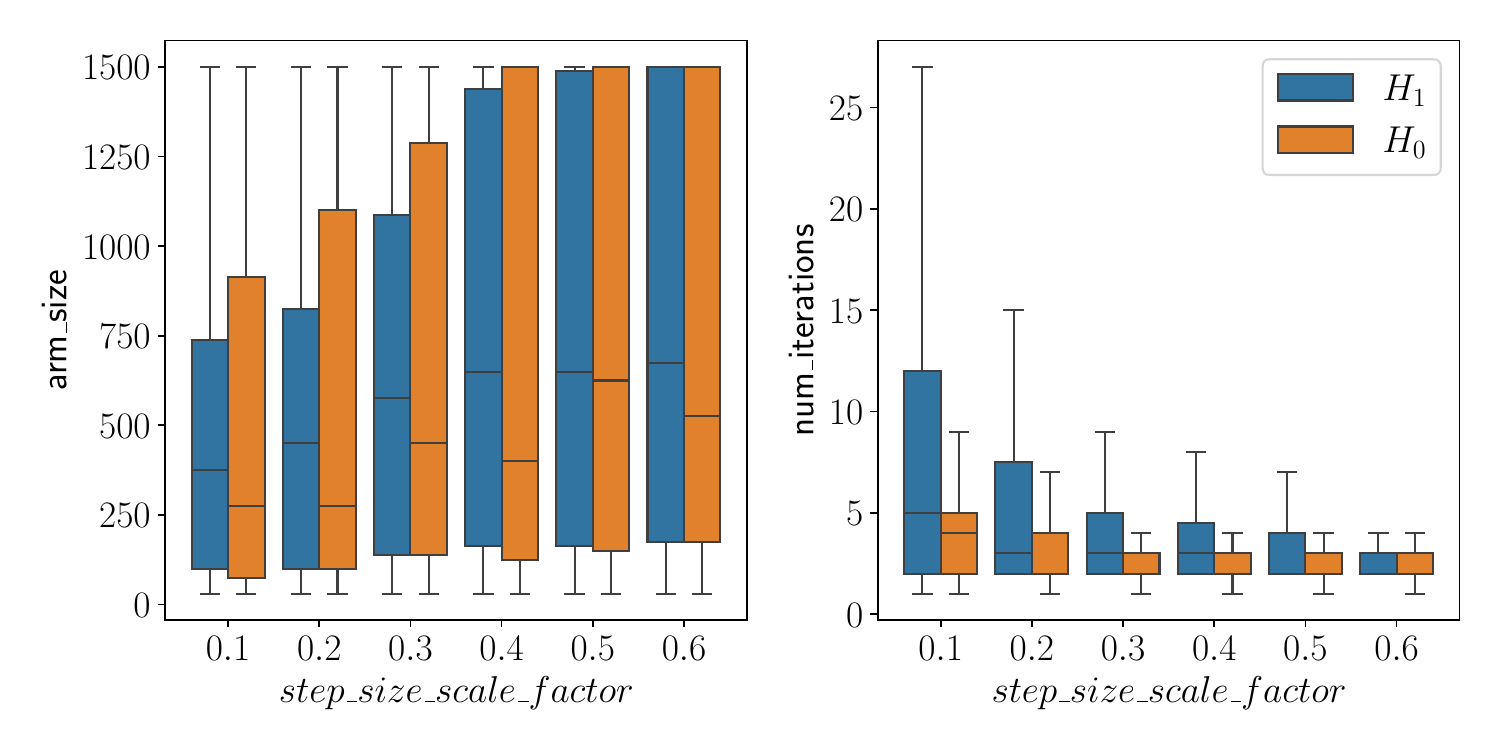}}
    \caption{Box plot over final arm size and number of iterations for TAD-SIE with $1-\beta_{target}=80\%$ and $\alpha_{target}=5\%$ across different values of \emph{step\_size\_scale\_factor} when either $H_1$ or $H_0$ holds.}
    \label{fig:sample_size_scale_factor_0.2}
\end{figure}

\begin{figure}[hbt!]
\makebox[\textwidth][c]{\includegraphics[width=1\textwidth]{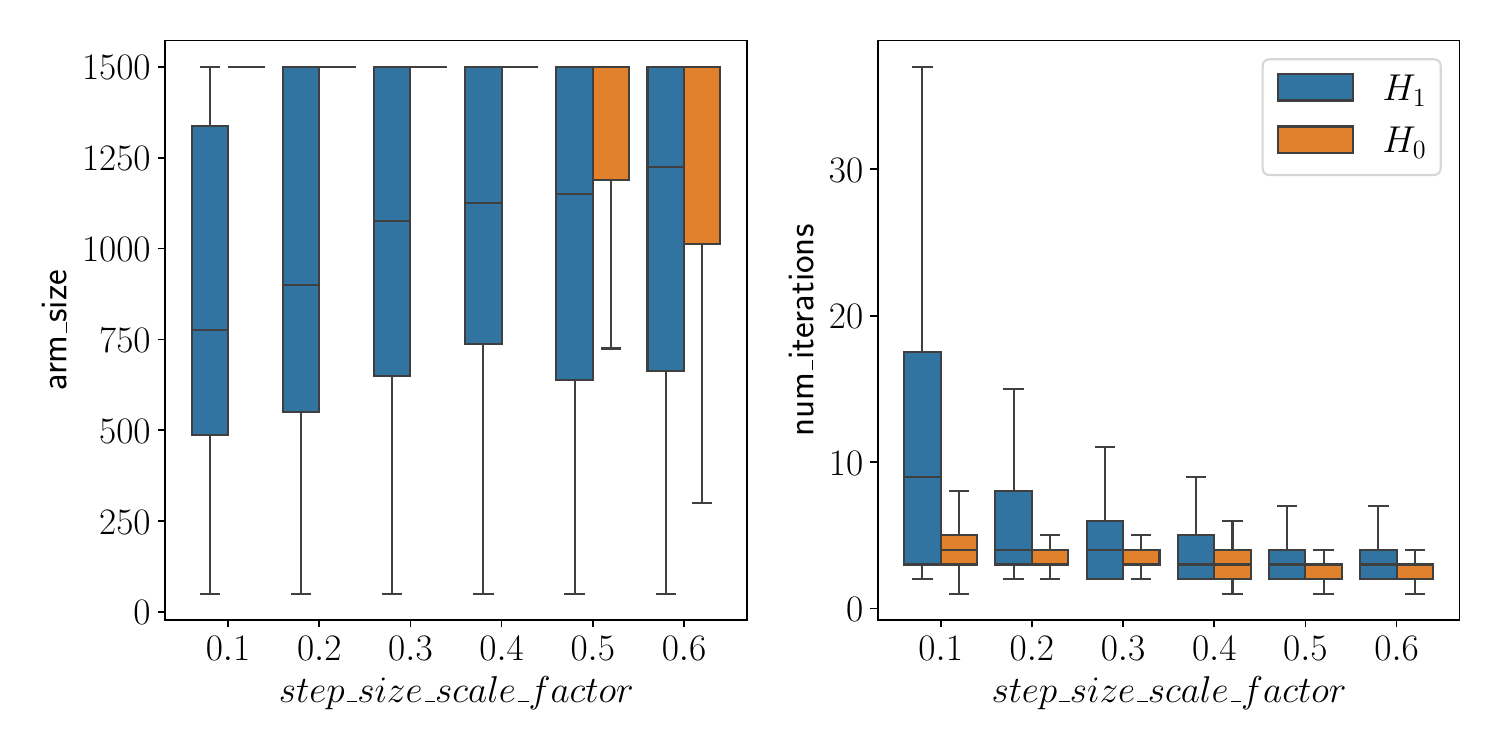}}
    \caption{Box plot over final arm size and number of iterations for TAD-SIE with $1-\beta_{target}=90\%$ and $\alpha_{target}=5\%$ across different values of \emph{step\_size\_scale\_factor} when either $H_1$ or $H_0$ holds.}
    \label{fig:sample_size_scale_factor_0.1}
\end{figure}

\subsection{TAD-SIE vs. Baselines}

Given hyperparameter values for TAD-SIE, i.e., TAD-SIE-SE and TAD-SIE-TE, we compare its performance against the baseline 
algorithms. 
TAD-SIE is able to reach desired performance under typical target operating points while the baselines cannot and obtain
low accuracies. Results across the algorithms for when target power is set to 80\% and target significance level
is set to 5\% are shown in Table~\ref{tab:baselines_0.8}. Fixed Sample Design results in lower power (48\%) and
higher significance level (9\%) since it uses noisy estimates obtained from a small pilot study to estimate the
required sample size. Standard-TAD improves the Fixed Sample Design by allowing the initial sample size
estimated under Fixed Sample Design to be increased based on interim data; however, the actual improvement is
marginal (i.e., power increases to 49\% while significance level decreases to 6\%) since few trials perform an
increase when $H_1$ holds given that Standard-TAD imposes a stringent condition for when the sample size can be
increased in order to control the significance level. Specifically, when $H_0$ holds, 7\% of trials meet the criterion and 4\% of trials increase the sample size, thereby preventing type-1 inflation; however, when $H_1$ holds, only 51\% of trials meet the criterion and 17\% of trials increase the sample size, thereby precluding gains in power. 

TAD-SIE improves upon Standard-TAD by integrating SECRETS to boost power and allowing trials to increase sample size 
without any condition while controlling the significance level with futility stopping. Specifically, both TAD-SIE-SE and TAD-SIE-TE have powers and significance levels that are comparable to the target operating point. The difference between the two solutions is that TAD-SIE-SE arrives at the target operating point with a lower sample size at the cost of more iterations while TAD-SIE-TE arrives at the target operating point in fewer iterations at the cost of higher sample size, as shown in Fig.~\ref{fig:se_vs_te_0.8}. For example, when $H_0$ holds, TAD-SIE-SE has a median arm size of 300 and requires a median of 4 iterations while TAD-SIE-TE has a median arm size of 538 and requires a median of 2 iterations. Similar trends appear when $H_1$ holds, although the median arm sizes are higher, i.e., 375 and 688 under TAD-SIE-SE and TAD-SIE-TE, respectively, and the median number of iterations is comparable or higher, i.e., 5 and 2 under TAD-SIE-SE and TAD-SIE-TE, respectively. The increase in arm size and number of iterations under $H_1$ is expected since futility stopping is less likely to be invoked. 
       
Similar trends hold when target power is 90\% and target significance level is 5\%, with TAD-SIE-SE and TAD-SIE-TE 
outperforming the baselines, though all methods show an expected increase in power, as shown in 
Table~\ref{tab:baselines_0.9}. Similar tradeoffs between sample size and number of iterations under TAD-SIE-SE and 
TAD-SIE-TE are also observed, as shown in Fig.~\ref{fig:se_vs_te_0.9}, though both modes show expected increases in median sample sizes and number of iterations. In addition, by aiming for a higher power, TAD-SIE-SE and TAD-SIE-TE rapidly reach the maximum arm size under $H_0$, hence incurring larger median arm sizes under $H_0$ than under $H_1$, which also leads to early termination, thereby keeping the number of iterations low under $H_0$.        

\begin{table*}[h!]
\centering
\caption{Performance under baseline algorithms and TAD-SIE with $1-\beta_{target}=80\%$ and $\alpha_{target}=5\%$.}
\label{tab:baselines_0.8} 
\resizebox{\columnwidth}{!}{%
\begin{tabular}{|M{0.33\linewidth} M{0.33\linewidth} M{0.33\linewidth}|}
\hline

\textbf{Method} & $\boldsymbol{1-\beta}$ \textbf{(\%)}  & $\boldsymbol{\alpha}$ \textbf{(\%)}  \\ \hline

Fixed Sample Design & 48 & 9\\ \hline 
Standard-TAD & 49 & 6 \\ \hline
TAD-SIE-SE & 81 & 4 \\ \hline 
TAD-SIE-TE & 80 & 4 \\ \hline 
\end{tabular}%
}
\end{table*}

\begin{figure}[hbt!]
\makebox[\textwidth][c]{\includegraphics[width=1\textwidth]{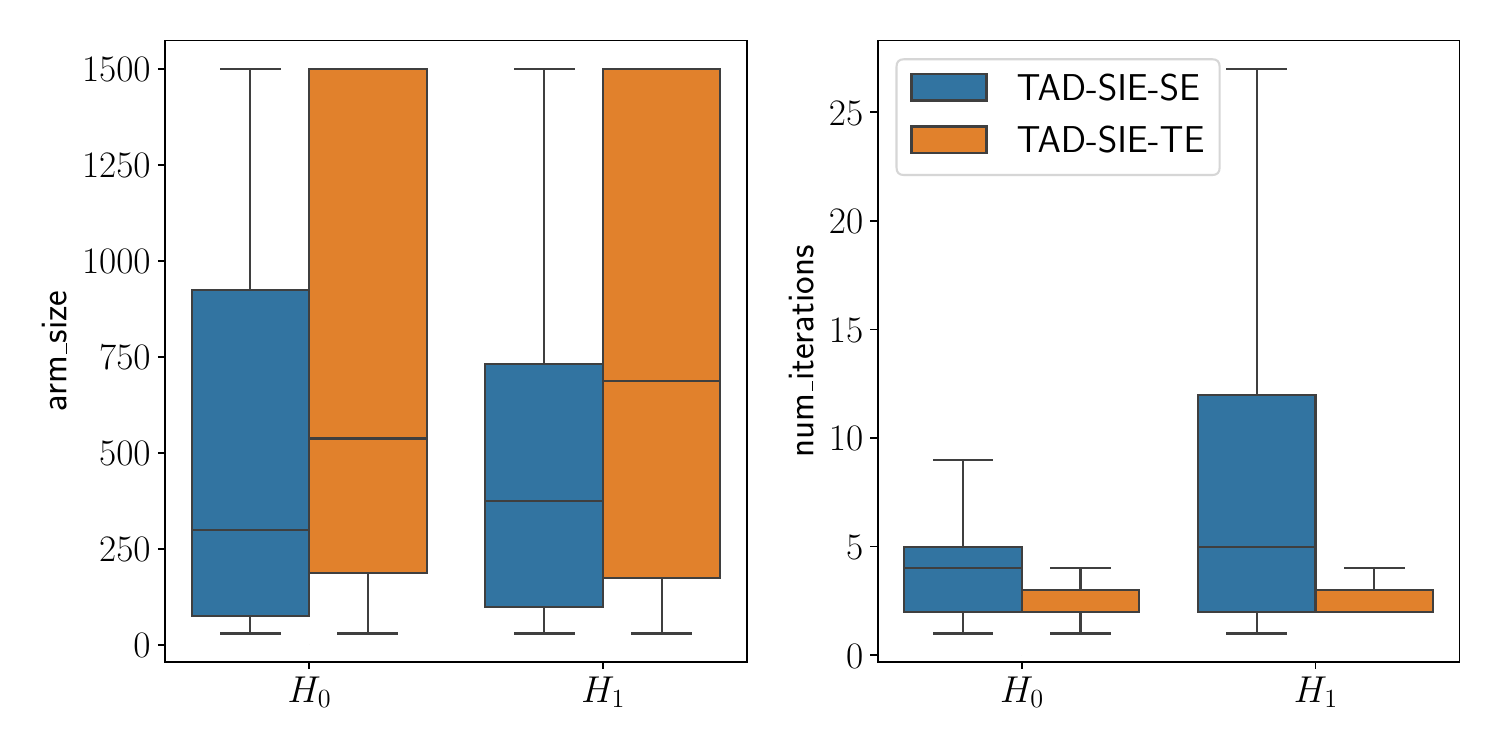}}
    \caption{Box plot over final arm size and number of iterations when either $H_0$ or $H_1$ holds for the sample-efficient and time-efficient modes of TAD-SIE with $1-\beta_{target}=80\%$ and $\alpha_{target}=5\%$.}
    \label{fig:se_vs_te_0.8}
\end{figure}

\begin{table*}[h!]
\centering
\caption{Performance under baseline algorithms and TAD-SIE when $1-\beta_{target}=90\%$ and $\alpha_{target}=5\%$.}
\label{tab:baselines_0.9} 
\resizebox{\columnwidth}{!}{%
\begin{tabular}{|M{0.33\linewidth} M{0.33\linewidth} M{0.33\linewidth}|}
\hline

\textbf{Method} & $\boldsymbol{1-\beta}$ \textbf{(\%)}  & $\boldsymbol{\alpha}$ \textbf{(\%)}  \\ \hline 

Fixed Sample Design & 57 & 8 \\ \hline 
Standard-TAD & 59 & 4 \\ \hline
TAD-SIE-SE & 98 & 6 \\ \hline 
TAD-SIE-TE & 91 & 4 \\ \hline 
\end{tabular}%
}
\end{table*}

\begin{figure}[hbt!]
\makebox[\textwidth][c]{\includegraphics[width=1\textwidth]{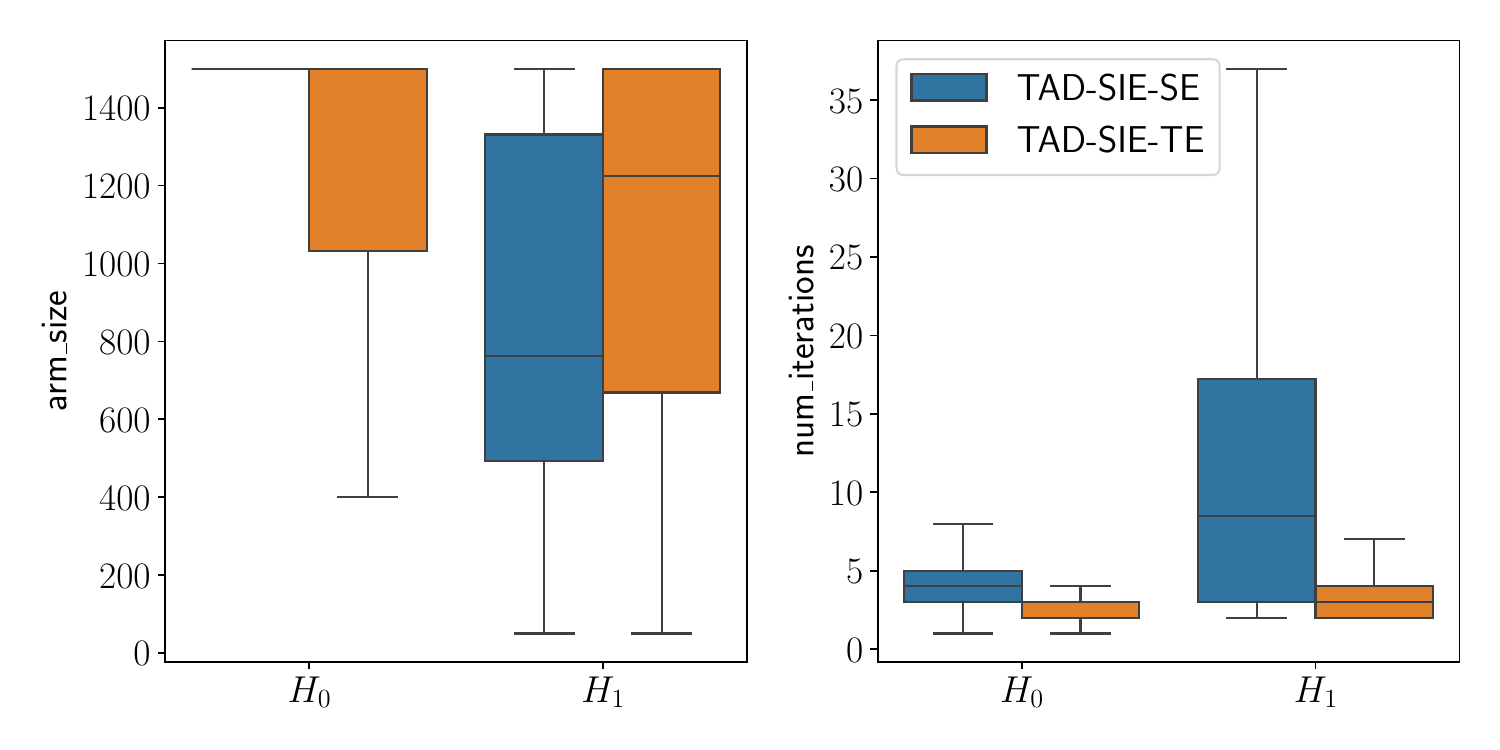}}
    \caption{Box plot over final arm size and number of iterations when either $H_0$ or $H_1$ holds for the sample-efficient and time-efficient modes of TAD-SIE with $1-\beta_{target}=90\%$ and $\alpha_{target}=5\%$.}
    \label{fig:se_vs_te_0.9}
\end{figure}


\subsection{Ablation Studies}

Having established TAD-SIE's effectiveness, we present results from the ablation studies to demonstrate the importance of each component underlying TAD-SIE. Specifically, we establish the importance of the moment estimation procedure, trend-adaptive algorithm, and hypothesis testing algorithm.      

\subsubsection{Moment Estimation Strategy} \label{sec:ablation:moment_est}

Our procedure for estimating moments, specifically the variance, is necessary for reducing bias in sample size
estimates under TAD-SIE. To demonstrate this, we swap the variance calculation step in
\emph{estimate\_moments} with a naive approach that estimates variance that defines the sample size formula
(Eq.~(\ref{eq:sample_size})) by variance over the ITEs. The naive approach can either underestimate the variance, thereby lowering the final sample size at the cost of reaching target power, or overestimate the variance, thereby increasing the final sample size at the cost of lowering the algorithm's efficiency in reaching the target power. Empirically, we find the former result does not hold as TAD-SIE with the naive moment estimation strategy is able to reach the target operating point of 80\% power and 5\% significance level. Therefore, we compare the two moment estimation strategies with respect to TAD-SIE's efficiency, i.e., the final sample size and number of iterations incurred; this is shown for TAD-SIE-SE in Fig.~\ref{fig:moment_est_arm_size_num_iter_sample_efficient}. When $H_0$ holds, TAD-SIE-SE with the naive approach has a median arm size of 463 and incurs a median of 3 iterations while TAD-SIE-SE with our approach has a median arm size of 300 and incurs a median of 4 iterations. When $H_1$ holds, TAD-SIE-SE with the naive approach has a median arm size of 600 and incurs a median of 8.5 iterations while TAD-SIE-SE with our approach has a median arm size of 375 and incurs a median of 5 iterations. Similar trends manifest for 
TAD-SIE-TE, as shown in Fig.~\ref{fig:moment_est_arm_size_num_iter_time_efficient}, although the number of iterations is comparably low across both approaches due to the rapid increase in the sample size over the iterations. The results imply that the naive moment estimation strategy overestimates the variance substantially compared to our moment estimation strategy, thereby increasing the final sample size and subsequently the number of iterations required by the algorithm to reach the target operating point. Our moment estimation strategy allows TAD-SIE to be more efficient.

\begin{figure}[hbt!]
\makebox[\textwidth][c]{\includegraphics[width=1\textwidth]{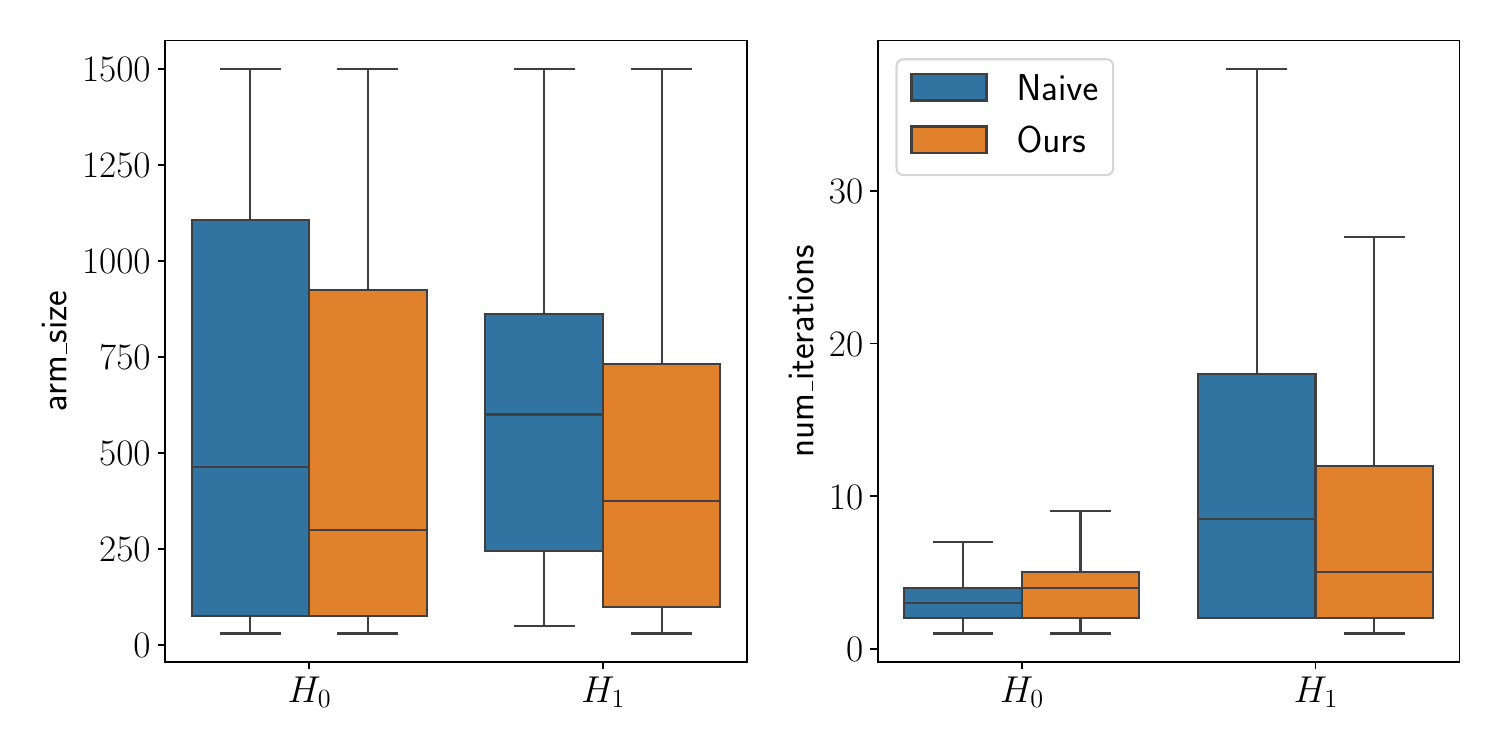}}
    \caption{Box plot over final arm size and number of iterations under the sample-efficient mode of TAD-SIE using the naive approach or our approach for estimating moments with $1-\beta_{target}=80\%$ and $\alpha_{target}=5\%$.}
    \label{fig:moment_est_arm_size_num_iter_sample_efficient}
\end{figure}

\begin{figure}[hbt!]
\makebox[\textwidth][c]{\includegraphics[width=1\textwidth]{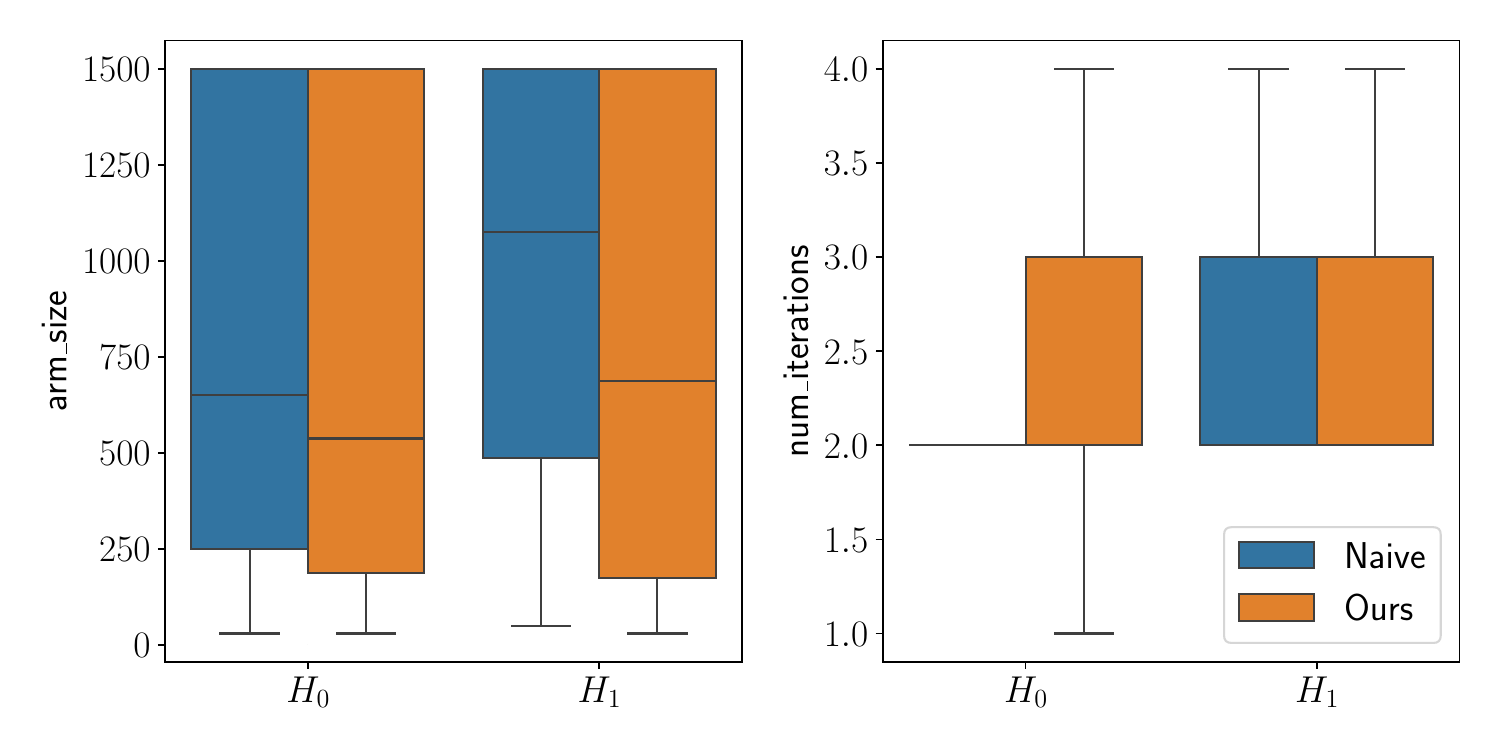}}
    \caption{Box plot over final arm size and number of iterations under the time-efficient mode of TAD-SIE using the naive approach or our approach for estimating moments with $1-\beta_{target}=80\%$ and $\alpha_{target}=5\%$.}
    \label{fig:moment_est_arm_size_num_iter_time_efficient}
\end{figure}

The moment estimation strategy also affects futility stopping by determining CP. To isolate its effect, we use our moment estimation strategy for sample size calculation in \emph{get\_step\_size} but vary the strategy used to calculate CP, i.e., the interim test statistic $z$ in \emph{check\_for\_futility} (line \ref{cff:z} of Alg.~\ref{alg:cff}). The results of this ablation for TAD-SIE-SE with target power 80\% and target significance level 5\% are shown in Fig.~\ref{fig:calc_cp_correct_se}. While significance levels remains comparable under the two approaches across values of \emph{futility\_power\_boundary}, power is lower under the ablated version of TAD-SIE that uses the naive moment estimation strategy for calculating CP. Since the naive approach overestimates variance, this shrinks the value of the interim test statistic and subsequently reduces CP (as seen in Fig.~\ref{fig:t_vs_cp}), thereby increasing the chance for futility stopping. The effect on futility stopping rates is negligible when $H_0$ holds since the magnitude of the ATE is already low and determines the futility stopping rate; e.g., when \emph{futility\_power\_boundary} is 1\%, futility stopping rates are 91\% and 89\% under our and the naive moment estimation strategies, respectively, and when \emph{futility\_power\_boundary} is at least 9\%, futility stopping rates are 96\% under both moment estimation strategies. However, when $H_1$ holds, the magnitude of the ATE is not low so the deflation of the interim test statistic induced under the naive moment estimation approach causes power to drop substantially. Similar trends manifest under TAD-SIE-TE with target power 80\% and target significance level 5\%, as shown in Fig.~\ref{fig:calc_cp_correct_te}. Therefore, our moment estimation strategy is essential for sustaining high power.             

\begin{figure}[hbt!]
\makebox[\textwidth][c]{\includegraphics[width=1\textwidth]{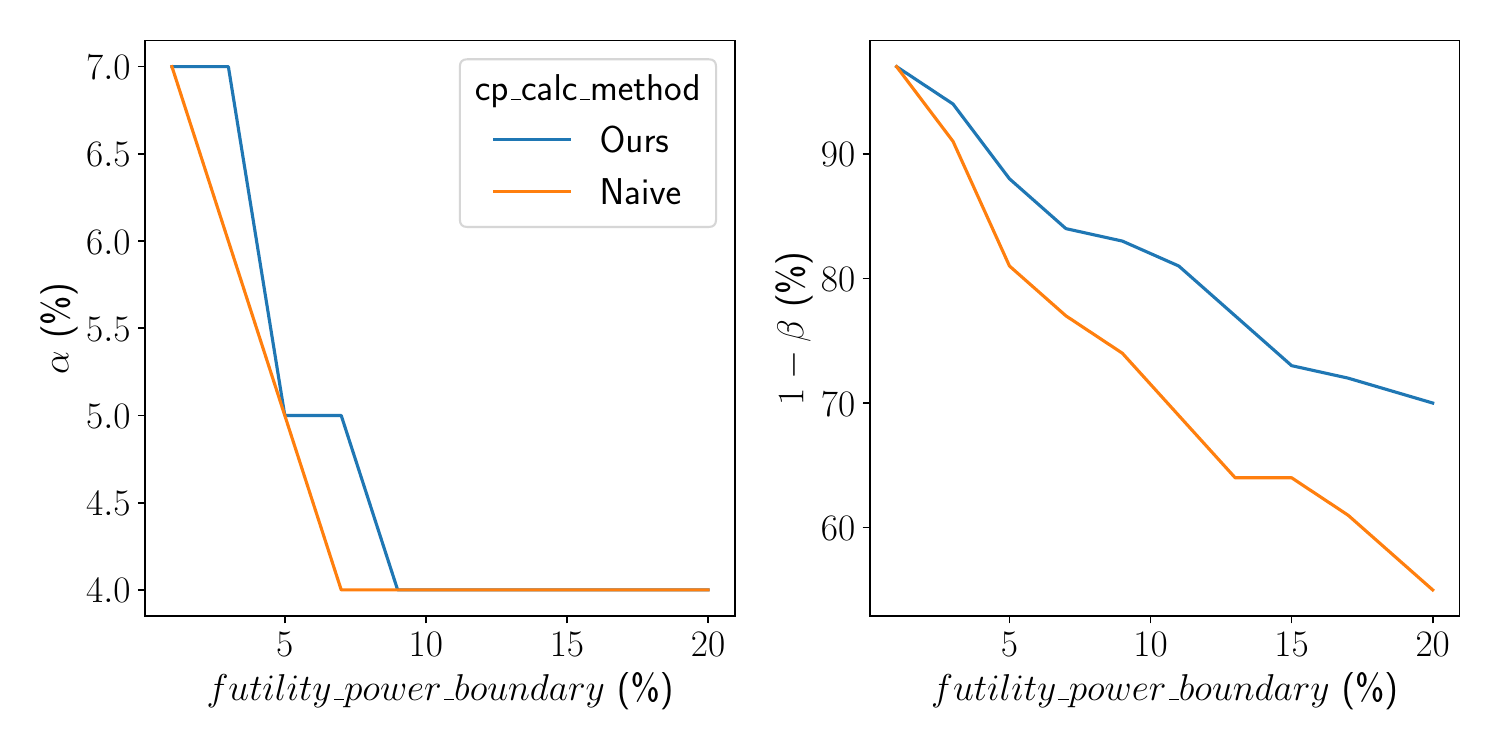}}
    \caption{Effect of moment estimation strategy used for calculating CP (``cp\_calc\_method") on significance level and power for the sample-efficient mode of TAD-SIE with $1-\beta_{target}=80\%$ and $\alpha_{target}=5\%$.}
    \label{fig:calc_cp_correct_se}
\end{figure}

\begin{figure}[hbt!]
\makebox[\textwidth][c]{\includegraphics[width=1\textwidth]{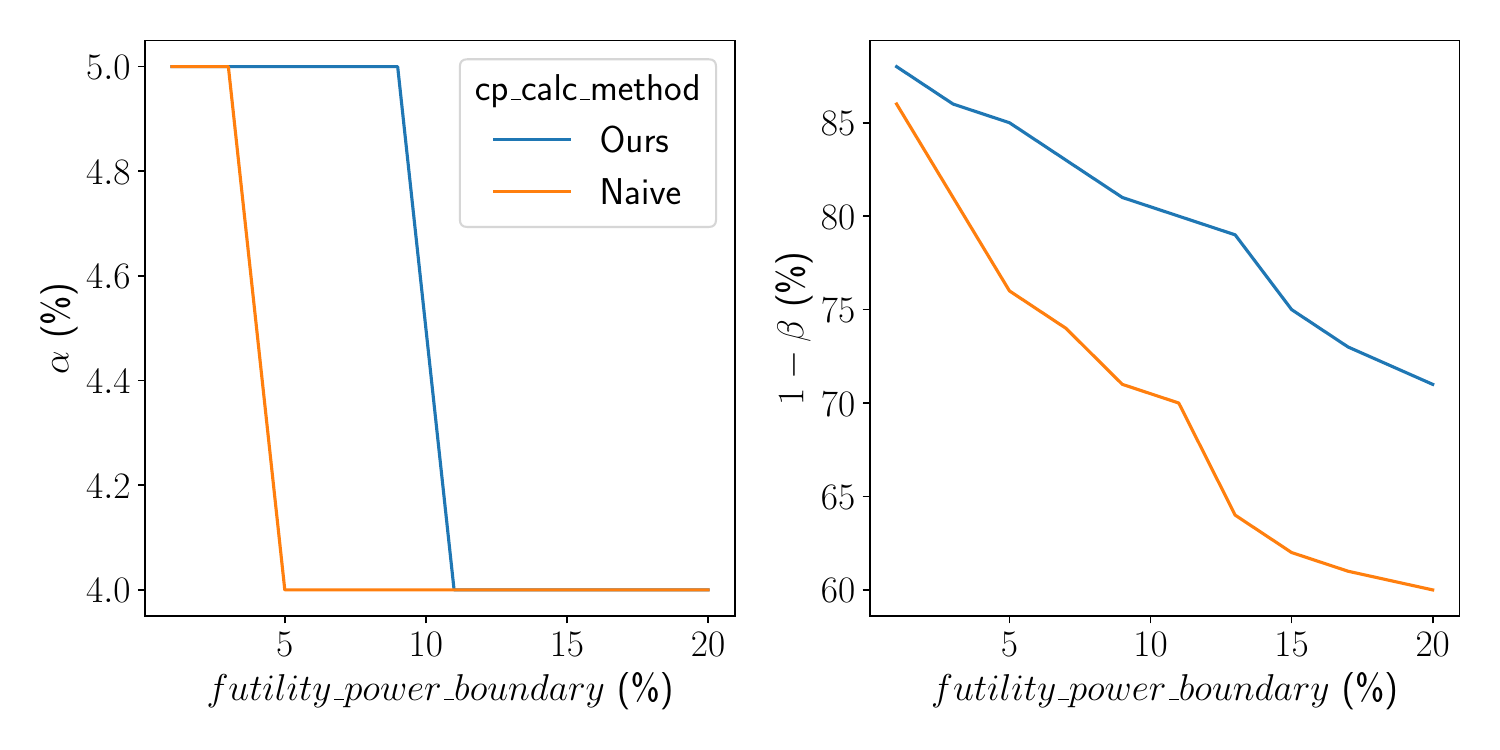}}
    \caption{Effect of moment estimation strategy used for calculating CP (``cp\_calc\_method") on significance level and power for the time-efficient mode of TAD-SIE with $1-\beta_{target}=80\%$ and $\alpha_{target}=5\%$.}
    \label{fig:calc_cp_correct_te}
\end{figure}

\subsubsection{trend-adaptive algorithm} 

Next, we show that our trend-adaptive algorithm is essential to TAD-SIE by swapping it with a standard trend-adaptive
algorithm \cite{chen2004increasing}, which we refer to as Standard-TAD+SIE; this ablated version of the algorithm is
similar to the baseline given by Standard-TAD but uses the proposed moment estimation procedure for estimating moments
and uses SECRETS for hypothesis testing. As opposed to TAD-SIE, Standard-TAD+SIE yields only one solution, i.e., there 
is no sample-efficient or time-efficient mode since the standard trend-adaptive algorithm simply increases the sample size 
based on the sample size formula without introducing a hyperparameter that scales it. 

Results for the ablation when the target power is 80\% and target significance level is 5\% are reported in
Table~\ref{tab:tad_ablation}. While all methods have significance levels close to 5\%, TAD-SIE-SE and TAD-SIE-TE
get powers close to 80\% while Standard-TAD+SIE only results in 62\% power. Power drops substantially under
Standard-TAD+SIE because the standard trend-adaptive algorithm imposes a stringent condition based on CP that
precludes trials from increasing their sample sizes after the initial sample size calculation in order to
control the significance level. In particular, under Standard-TAD+SIE, only 4\% of trials increase the sample size when $H_0$ holds, thereby controlling significance level; however, this comes at the cost of power since only 13\% of trials increase the sample size when $H_1$ holds. In contrast, our trend-adaptive design allows trials to increase sample sizes without any condition while monitoring futility, thereby allowing TAD-SIE to reach target power and control 
the significance level. In particular, under TAD-SIE-SE, when $H_0$ holds, 71\% of trials perform at least one sample size increase but 96\% of trials terminate by futility; when $H_1$ holds, 73\% of trials perform at least one sample size increase and only 19\% of trials terminate by futility. Similar trends manifest under TAD-SIE-TE, although the fraction of trials 
that perform at least one increase declines since TAD-SIE-TE uses larger step sizes. These results establish the effectiveness of our trend-adaptive algorithm in reaching the target operating point. 

\begin{table*}[hbt!]
\centering
\caption{TAD-SIE (sample-efficient and time-efficient modes) compared against an ablated version of TAD-SIE that swaps the proposed trend-adaptive algorithm for a standard one \cite{chen2004increasing}, with $1-\beta_{target}=80\%$ and $\alpha_{target}=5\%$. For each method, we report significance level and power, percent of trials that increase the sample size at least once when $H_0$ or $H_1$ holds, and percent of trials that terminate by futility when $H_0$ or $H_1$ holds, given by the columns, respectively. For Standard-TAD+SIE, the entry under ``\% terminating by futility" is n/a since this algorithm does not use futility stopping.}
\label{tab:tad_ablation} 
\resizebox{\columnwidth}{!}{%
\begin{tabular}{|M{0.25\linewidth} M{0.4\linewidth} M{0.25\linewidth} M{0.25\linewidth}|}
\hline
\textbf{Method} & $\boldsymbol{\alpha \ / \ 1-\beta}$ \textbf{(\%)} & \textbf{\% performing increase} ($\boldsymbol{H_{0}}/\boldsymbol{H_{1}}$) & \textbf{\% futility stopping} ($\boldsymbol{H_{0}}/\boldsymbol{H_{1}}$)  \\ \hline 
TAD-SIE-SE & 4/81 & 71/73 & 96/19 \\ \hline  
TAD-SIE-TE & 4/80 & 27/42 & 96/20 \\ \hline 
Standard-TAD+SIE & 6/62 & 4/13 & n/a \\ \hline
\end{tabular}%
}
\end{table*}   

%
%

\subsubsection{SECRETS vs. Standard Hypothesis Testing}

SECRETS is vital to the performance gains under TAD-SIE. To demonstrate this, we appropriately modify TAD-SIE to
make it suitable for standard hypothesis testing. Specifically, we modify \emph{estimate\_moments} to
calculate the ATE and variances of the outcome over the control and treatment groups based on the original RCT
data; these equations are captured in Eq.~(\ref{eq:moments_two_sample}), where $o_{control}$ and $o_{treat}$ are
vectors that contain primary outcomes per subject in the control and treatment groups, respectively. We also
modify \emph{get\_step\_size} so that it takes in the variance for the control and treatment groups and uses
the appropriate sample size formula, i.e., Eq.~(\ref{eq:sample_size_two_sample}), in 
line~\ref{agss:get_target} of Alg.~\ref{alg:gss}; we do not divide by two since
Eq.~(\ref{eq:sample_size_two_sample}) returns the sample size per arm. We modify \emph{check\_for\_futility} similarly so that it takes in the two variance terms and uses the appropriate test statistic formula, i.e., 
Eq.~(\ref{eq:test_stat_two_sample}) \cite{rosner2015fundamentals}, in line~\ref{cff:z} of Alg.~\ref{cff:cp}. Finally, we switch the hypothesis testing procedure from SECRETS to the two-sample \emph{t}-test for independent samples with unequal variances \cite{rosner2015fundamentals}.

\begin{align}
\label{eq:moments_two_sample}
&\delta =\text{Mean}(o_{treat})-\text{Mean}(o_{control}) \\ 
&\sigma_{ctrl}^2 =\text{Var}(o_{control}) \nonumber \\ 
&\sigma_{treat}^2 =\text{Var}(o_{treat}) \nonumber  
\end{align}  

\begin{equation}
\label{eq:test_stat_two_sample}
z=\frac{\delta}{\sqrt{(\sigma_{control}^2 + \sigma_{treat}^2)/n_a}}
\end{equation} 

Results that compare TAD-SIE against the ablated version with target power of 80\% and target significance level of 5\% are 
shown in Table~\ref{tab:hyp_test}. Specifically, for TAD-SIE-SE, swapping TAD-SIE with the version designed for standard hypothesis testing reduces power from 81\% to 22\% even though significance level increases from 4\% to 5\%; similar trends manifest for TAD-SIE-TE although the drop in power is not as large since the algorithm terminates with a larger sample size (i.e., when $H_1$ holds, the median final arm size is 75 under TAD-SIE-SE and 450 under TAD-SIE-TE). TAD-SIE designed for SECRETS is superior to the ablated version because SECRETS is able to boost power over standard hypothesis testing at any arm size by simulating a cross-over trial, thereby increasing the chance of reaching target power within the trial's sample size constraints.  

\begin{table*}[hbt!]
\centering
\caption{TAD-SIE (sample-efficient and time-efficient modes) compared against a version of TAD-SIE designed for two-sample \emph{t}-testing with $1-\beta_{target}=80\%$ and $\alpha_{target}=5\%$. The version is indicated by the ``Hypothesis Testing Strategy" column. 
}
\label{tab:hyp_test} 
\resizebox{\columnwidth}{!}{%
\begin{tabular}{|M{0.25\linewidth} M{0.4\linewidth} M{0.25\linewidth} M{0.25\linewidth}|}
\hline

\textbf{Method} & \textbf{Hypothesis Testing Strategy} & $\boldsymbol{1-\beta}$ \textbf{(\%)} & $\boldsymbol{\alpha}$ \textbf{(\%)}  \\
\hline

\multirow[c]{2}{*}{TAD-SIE-SE} & SECRETS & 81 & 4   \\ 
& Standard & 22 & 5 \\  \hline 

\multirow[c]{2}{*}{TAD-SIE-TE} & SECRETS & 80 & 4   \\  
& Standard & 54 & 5\\  \hline   

\end{tabular}%
}
\end{table*}

\section{Discussion} \label{sec:discuss}

Our results show that TAD-SIE is able to reach target accuracies, given poor initial sample size estimates for 
a parallel-group RCT, by using a trend-adaptive algorithm designed for a powerful hypothesis testing strategy.
Specifically, it leverages SECRETS for hypothesis testing, which offers substantial boosts in power at any given arm
compared to standard hypothesis testing; hence, integrating this hypothesis testing algorithm into an adaptive design 
makes it easier to reach target powers within a trial's constraints, especially under noisy estimates of the ATE and 
outcome variance, as demonstrated by our ablation that adapts TAD-SIE to standard hypothesis testing. However, reaching 
the target operating point under SECRETS requires a new TAD since SECRETS introduces dependencies among samples, a 
property that violates assumptions under standard TADs. Specifically, our ablations show that our procedure for estimating 
the variance properly accounts for dependencies among the samples so that sample size estimates and CP calculations for 
futility stopping are more accurate. In addition, since the standard TAD is ineffective at reaching target power, our TAD 
overcomes this limitation by removing the sample increase criterion and iteratively increasing the sample size based on 
improved sample size estimates, while controlling the significance level with futility stopping.

In addition, TAD-SIE can toggle between sample-efficient or time-efficient modes to realize the target operating point, 
thereby providing an additional degree of freedom to the designer that has not been considered before in the context of 
adaptive designs as two-stage adaptive designs are common \cite{mehta2011adaptive}. TAD-SIE implements these modes by 
decreasing or increasing \emph{sample\_size\_scale\_factor}, which determines the rate at which the sample size is 
increased at each iteration, while adopting an adequate futility stopping threshold; based on results from a sample RCT, 
we suggested default settings for these hyperparameters.

While TAD-SIE was presented in the context of a standard parallel-group RCT, i.e., a two-arm superiority trial, it is 
applicable to multi-arm studies, since they reduce to pairwise comparisons between select arms 
\cite{juszczak2019reporting}, and to studies where equivalence or noninferiority is being assessed 
\cite{friedman2015fundamentals} since these only affect the test strategy used and require appropriate changes to 
calculations of sample size, test statistic, and CP within TAD-SIE (including SECRETS) \cite{chow2002note}.   

While we have demonstrated TAD-SIE's effectiveness on a sample RCT, TAD-SIE inherits the limitations of adaptive
designs, namely that it assumes that the primary outcome is rapidly measurable, a property that does not hold for all
RCTs (e.g., RCTs involving chronic conditions) \cite{friedman2015fundamentals}. If the primary outcome takes too long 
to measure, it may not be possible to perform an adaptation, even under the time-efficient mode of TAD-SIE, since the trial gets prolonged, making it obsolete \cite{friedman2015fundamentals,wason2019keep}. On a related note, the sample-efficient mode of TAD-SIE assumes that the enrollment rate is slower than the rate at which the primary outcome is measured, a constraint 
that sometimes may not be feasible to satisfy since this can also prolong the trial, given that TAD-SIE-SE can require 
more than two iterations. If the enrollment rate is faster, however, TAD-SIE-SE cannot guarantee sample size reductions since by the time the algorithm terminates with the true final sample size, the number of recruited participants could exceed that value \cite{mehta2011adaptive}. In future work, we plan to modify TAD-SIE so that it can satisfy constraints on response time and enrollment rates by having it estimate the primary outcome from correlated, rapidly measurable outcomes rather than directly measuring its value \cite{wason2019keep}, given that trials collect data across other response variables \cite{friedman2015fundamentals,sertkaya2016key}. We will look into schemes for fitting such an estimator from data obtained from the internal pilot study and accrued throughout the trial.

\section{Conclusion} \label{sec:con}

In conclusion, we presented TAD-SIE, a novel framework that implements a parallel-group RCT to reach a target power and 
target significance level in the absence of reliable sample size estimates obtained from study planning. To do this, 
TAD-SIE implements a trend-adaptive algorithm designed for a state-of-the-art hypothesis testing strategy and one that 
addresses limitations of a standard TAD in order to increase the chance of reaching target power within a trial's 
constraints. Since TAD-SIE is a sequential algorithm, we suggest hyperparameters to implement sample-efficient or 
time-efficient modes to accommodate different resource constraints. Evaluated on a real-world Phase-3 clinical 
parallel-group RCT under a common design (i.e., two-arm superiority trial with an equal number of subjects per arm), 
TAD-SIE reaches typical target operating points of 80\% or 90\% power and 5\% significance level, unlike standard 
approaches. In particular, for a target power of 80\% and 5\% significance level, the sample-efficient mode of TAD-SIE 
requires a median of 300-400 subjects per arm and a median of 4-5 iterations while the time-efficient mode of TAD-SIE 
requires a median of 500-700 subjects per arm and a median of 2 iterations, showing that TAD-SIE can be used to 
effectively trade off resources (similar trends were observed for a target power of 90\% and target significance level 
of 5\%). Our results demonstrate that TAD-SIE can be used to increase the success rate for Phase-3 clinical trials, 
thereby increasing efficiency of the drug development cycle and saving hundreds of millions of USD per drug. However, 
since TAD-SIE is iterative, it is only practical for RCTs where the time to measure the primary outcome is small. Our 
future work will extend TAD-SIE to settings where these times are longer.    

\section*{Acknowledgment}  

This work was supported by NSF under Grant No. CNS-1907381 and was performed using resources from Princeton Research 
Computing. 

This research is based on data from NINDS obtained from its Archived Clinical Research Dataset website. The CHAMP dataset 
was obtained from the Childhood and Adolescent Migraine Prevention Study, conducted under principal investigators (PIs) 
Drs. Powers, Hershey, and Coffey, under Grant \#1U01NS076788-01.        


\vskip 0.2in
\bibliography{ref}
\bibliographystyle{IEEEtran} 

\newtheorem{thm}{Theorem}[subsection]
\renewcommand{\thethm}{\arabic{thm}}

\section*{Appendix} \label{sec:appendix}

\begin{thm} 

The variance calculated by the procedure outlined in estimate\_moments scales with sample size $n$, i.e., $\sigma^{2}=O(n)$. 

\end{thm}

\begin{proof} \label{proof:thm_var_scale}

Let $\{X_i\}_{i=1}^{n}$ be a set of random variables that could be dependent with $\mathbb{E}[X_i]=0$, and let $\bar{X}\coloneqq\frac{1}{n}{ \sum\limits_{i=1}^{n} X_i }$. First, we establish bounds on $\sigma_{\bar{X}}^{2}$, the variance of the mean, given by 

\begin{equation} \label{eq:var_mean}
\sigma_{\bar{X}}^{2} =\text{Var}\Bigg( \frac{\sum_{i=1}^{n} X_i }{n}  \Bigg) \\ 
\end{equation}

\noindent
To do this, we establish bounds on the variance of the sum, shown in Eq.~(\ref{eq:var_sum_bound}). 

\begin{align} \label{eq:var_sum_bound}
\text{Var}\Bigg(\sum\limits_{i=1}^{n} X_i \Bigg) &=\mathbb{E}\Bigg[\Bigg(\sum\limits_{i=1}^{n}X_i\Bigg)^2\Bigg] \\
&= \sum\limits_{(i,j) \in [1,n] \times [1,n]} \mathbb{E}[X_iX_j] \nonumber \\
&= \begin{cases}
\theta(n), & \text{if}\ \{X_i\}_{i=1}^{n} \ \text{are independent} \nonumber \\
O(n^2), & \text{otherwise}
\end{cases}
\end{align}

The bound when $\{X_i\}_{i=1}^n$ are independent follows because there are $n$ terms in the summand since the cross-terms reduce to 0 given $\mathbb{E}[X_i]=0$. When $\{X_i\}_{i=1}^n$ are not independent, there can be at most $\frac{n(n-1)}{2}$ additional non-zero cross-terms in the summand, hence the sum is bounded above by $O(n^2)$.

Plugging these bounds into Eq.~(\ref{eq:var_mean}) establishes bounds on $\sigma_{\bar{X}}^{2}$, as shown in 
Eq.~(\ref{eq:var_mean_bound}).

\begin{align} \label{eq:var_mean_bound}
\sigma_{\bar{X}}^{2} &=\text{Var}\Bigg( \frac{\sum_{i=1}^{n} X_i }{n}  \Bigg) \\
&=\begin{cases}
      \theta(1/n), & \text{if}\ \{X_i\}_{i=1}^{n} \ \text{are independent} \nonumber \\
      O(1), & \text{otherwise}
    \end{cases}  
\end{align}  

Now consider another set of random variables $\{\tilde{X}_i\}_{i=1}^{n}$ that are i.i.d. with variance
$\sigma_{\tilde{X}}^2$, for which the following relationship holds: 

\begin{equation} \label{eq:var_iid_var_mean}
\text{Var}\Bigg( \frac{\sum_{i=1}^{n} \tilde{X}_i }{n}  \Bigg)=\sigma_{\bar{X}}^{2} \\
\end{equation}

Then, given the relationship between the variance of the mean and variance of i.i.d. samples from 
Eq.~(\ref{eq:var_mean_sample}), we plug in the bounds from Eq.~(\ref{eq:var_mean_bound}) to establish bounds on
$\sigma_{\tilde{X}}^{2}$, as shown in Eq.~(\ref{eq:var_iid_bound}):

\begin{align} \label{eq:var_iid_bound}
\sigma_{\tilde{X}}^{2}&=n\sigma_{\bar{X}}^2 \\
&=\begin{cases}
\theta(1), & \text{if}\ \{X_i\}_{i=1}^{n} \ \text{are independent} \nonumber \\
O(n), & \text{otherwise}
\end{cases}
\end{align}

\end{proof}
\end{document}